\documentclass{article}

\usepackage{amssymb}								% Lettres grecques
\usepackage{amsmath}								% Maths
\usepackage{eufrak}
\usepackage{pifont}
\usepackage{amsthm}											%Theorems
\usepackage{braket}											% Notation Braket Dirac
\usepackage{amsfonts}										% Fonts pour les maths
\usepackage{calc}												% Maths, op??rations
\usepackage{amssymb,stmaryrd,gensymb}		% Symbole pour les maths	%%
\usepackage{bbm}												% Fonts pour les maths
\usepackage{mathrsfs}										% Fonts pour les maths
\usepackage{graphicx}										% Graphiques
\usepackage[usenames,dvipsnames,svgnames,table]{xcolor} 	%Couleurs
\usepackage[breaklinks]{hyperref}				%Liens PDF
\usepackage[normalem]{ulem}
\usepackage[T1]{fontenc}									% Fran??ais
\usepackage[utf8]{inputenc}								% Encodage
\usepackage{pict2e}
\usepackage{tikz}
\usepackage[active]{srcltx}
\usepackage{caption}
\usepackage{ytableau}

\newcommand{\sTD}[1]{\overline{\underline{#1}}} 
\newcommand{\sT}[1]{\overline{#1}}
\newcommand{\sD}[1]{\underline{#1}}
\newcommand{\sTDL}{\,\overbar{\sD{\Lambda}}}
\newcommand{\sTL}{{\,\vphantom{\sTDL}\overbar{\Lambda}}}
\newcommand{\sDL}{{\,\vphantom{\sTDL}\sD{\Lambda}}}
\newcommand{\sLs}{ { {\,\vphantom{\sTDL}\Lambda} }^{s} }

\newcommand{\PrescJack}{\mathscr{E}}
\newcommand{\prodA}[1]{\langle #1\rangle}
\newcommand{\prodC}[1]{\langle\!\langle #1\rangle\!\rangle}
\newcommand{\tabsmall}[1]{\scalebox{.28}{$#1$}}

\newcommand{\sTrig}{
	\begin{tikzpicture}[scale=.14, remember picture]%[rotate=180]
	\begin{scope}[overlay]
		\draw[line width=0.4pt] (0,.5) circle [radius=1];
		\clip (0,.5) circle [radius=1];
		\draw[line width=0.4pt] (0,-1.5) -- (0,1.5);
	\end{scope}
	\end{tikzpicture}
}
\newcommand{\sCercle}{
	\begin{tikzpicture}[scale=.14, remember picture]%[rotate=180]
	\begin{scope}[overlay]
		\draw[line width=0.4pt] (0,.5) circle [radius=1];
		\clip (0,.5) circle [radius=1];
		\draw[line width=0.4pt] (-1,.5) -- (1,.5);
	\end{scope}
	\end{tikzpicture}
}

\newcommand{\sDouble}{
	\begin{tikzpicture}[scale=.14, remember picture]%[rotate=180]
	\begin{scope}[overlay]
		\draw[line width=0.4pt] (0,.5) circle [radius=1];
		\clip (0,.5) circle [radius=1];
		\draw[line width=0.4pt] (0,-1.5) -- (0,1.5);
		\draw[line width=0.4pt] (-1,.5) -- (1,.5);
	\end{scope}
	\end{tikzpicture}
}

\newcommand{\yT}{\none[\sTrig]}
\newcommand{\yC}{\none[\sCercle]}
\newcommand{\yTC}{\none[\sDouble]}

\newcommand{\superYsmall}[1]{\, \scalebox{.5}{\begin{ytableau}#1 \end{ytableau} }}
\newcommand{\superY}[1]{\begin{array}{c} \scalebox{1}{\begin{ytableau}#1 \end{ytableau} } \end{array}}
\newcommand{\superYlarge}[1]{\begin{array}{c} \scalebox{1.5}{\begin{ytableau}#1 \end{ytableau} } \end{array}}
\newcommand{\superYbig}[1]{\begin{array}{c} \scalebox{2}{\begin{ytableau}#1 \end{ytableau} } \end{array}}
\newcommand{\superYBig}[1]{\begin{array}{c} \scalebox{2.5}{\begin{ytableau}#1 \end{ytableau} } \end{array}}

\newcommand{\N}[1]{$\mathcal{N}=#1$}
\newcommand{\overbar}[1]{\mkern 1.5mu\overline{\mkern-1.5mu#1\mkern-1.5mu}\mkern 1.5mu}

\newtheorem{theorem}{Theorem}
\newtheorem{corollary}[theorem]{Corollary}
\newtheorem{conjecture}[theorem]{Conjecture}
\newtheorem{claim}[theorem]{Claim}
\newtheorem{proposition}[theorem]{Proposition}
\newtheorem{lemma}[theorem]{Lemma}
\theoremstyle{definition}
\newtheorem{example}[theorem]{Example}
\newtheorem{definition}[theorem]{Definition}
\theoremstyle{remark}

\let\la\lambda
\let\La\Lambda
\let\a\alpha

\def\C{{\mathbb C}}
\def\Nm{{\mathcal N}}

\let\d\partial

\let\ta\theta
\def\beq{\begin{equation}}
\def\eeq{\end{equation}}

\newcommand{\K}{\mathcal{K}}
\newcommand{\Dk}{\mathcal{D}}
\newcommand{\TDm}{ {\sTD{m}} }

\newcommand{\Pm}{\mathcal{P}}

\usepackage{enumitem,amssymb}
\newlist{todolist}{itemize}{2}
\setlist[todolist]{label=$\square$}
\usepackage{pifont}
\newcommand{\cmark}{\ding{51}}%
\newcommand{\xmark}{\ding{55}}%
\newcommand{\done}{\rlap{$\square$}{\raisebox{2pt}{\large\hspace{1pt}\cmark}}%
\hspace{-2.5pt}}
\newcommand{\wontfix}{\rlap{$\square$}{\large\hspace{1pt}\xmark}}
\ytableausetup{smalltableaux}
\title{The $\mathcal{N}=  2$ supersymmetric Calogero-Sutherland model and its eigenfunctions}
	\author{
		L. Alarie-V\'ezina,\thanks{D\'epartement de physique, de g\'enie physique et
  	d'optique, Universit\'e Laval,  Qu\'ebec, Canada,  G1V 0A6; ludovic.alarie-vezina.1@ulaval.ca}	 	L. Lapointe\thanks{Instituto de Matem\'atica y F\'{\i}sica, Universidad de
	  Talca, 2 Norte 685, Talca, Chile; lapointe@inst-mat.utalca.cl } and
	  P. Mathieu\thanks{D\'epartement de physique, de g\'enie physique et
	  d'optique, Universit\'e Laval,  Qu\'ebec, Canada,  G1V 0A6; pmathieu@phy.ulaval.ca}
	}

\begin{document}

		\maketitle
	 \begin{abstract}

In a recent work, we have initiated the theory of ${\mathcal N}=2$ symmetric superpolynomials. As far as the classical bases are concerned, this is a rather straightforward generalization of the ${\mathcal N}=1$ case. However this construction could not be generalized to the formulation of Jack superpolynomials. The origin of this obstruction is unraveled here, opening the path for building the desired Jack extension.  Those are shown to be obtained from the non-symmetric Jack polynomials by a suitable symmetrization procedure and an appropriate dressing by the anticommuting variables. This construction  is substantiated by the characterization of the ${\mathcal N}=2$ Jack superpolynomials as the eigenfunctions of the ${\mathcal N}=2$ supersymmetric version of the Calogero-Sutherland model, for which, as a side result, we demonstrate the complete integrability by displaying the explicit form of four towers of mutually commuting (bosonic) conserved quantities.  The ${\mathcal N}=2$ Jack superpolynomials are  orthogonal with respect to the analytical scalar product (induced by the quantum-mechanical formulation) as well as a new combinatorial scalar product defined on a suitable deformation of the power-sum basis.
	 \end{abstract}

	\tableofcontents
	\newpage

\section{Introduction}

The study of the ${\mathcal N}=2$ symmetric superpolynomials has been initiated in
\cite{Alarie-Vezina2015a}. Let us review briefly what is meant by this program.\\

The construction amounts to extending the classical symmetric polynomials
to functions depending on not only $x_1,\dots, x_N$ but also on two extra independent sets of anticommuting variables $\theta_1,\ldots ,\theta_N$ and $\phi_1,\ldots,\phi_N$. We require
 the variables in each set to anticommute among themselves:
\begin{equation} \theta_i\theta_j=-\theta_j\theta_i\,,\qquad \phi_i\phi_j=-\phi_j\phi_i,\end{equation} %require the anticommutativity of two of the
and also with each other:
 \begin{equation}\label{tap}\theta_i\phi_j=-\phi_j\theta_i.\end{equation}
Equivalently, we consider the ring $\C(x_1,\ldots,x_N)\otimes\bigwedge\bigl(\mathbb C(\ta_1,\ldots,\ta_n,\phi_1,\ldots, \phi_N) \bigr)$.
This addition of the anticommuting variables is understood in the context of superspace: the variables $\phi_i$ and $\ta_i$ are attached  to the bosonic variable $x_i$. Therefore, the symmetry requirement imposed on polynomials is the invariance under the interchange of two triplets $(x_i,\phi_i,\theta_i)\leftrightarrow (x_{\sigma(i)},\phi_{\sigma(i)},\theta_{\sigma(i)})$ where $\sigma$ belongs to $S_N$, the symmetric group on $N$ elements.  We call the resulting objects  ${\mathcal N}=2$ symmetric superpolynomials and denote their ring as
$\Pi^N$.\\

A detailed analysis of the ${\mathcal N}=2$ supersymmetric version of the classical bases $m_\lambda, e_\lambda, h_\lambda$ and $p_\lambda$ ($\lambda$ being a partition), was presented in \cite{Alarie-Vezina2015a}. Take for instance the power-sum basis. It is a multiplicative basis built out of four constituents:
\beq
p_n=\sum_{i=1}^N x_i^n,\qquad \sT{p}_r=\sum_{i=1}^N \phi_i\,x_i^r,\qquad \sD{p}_r=\sum_{i=1}^N\ta_i\, x_i^r,\qquad \sTD{p}_r=\sum_{i=1}^N\phi_i\ta_i\, x_i^r
\eeq
with $n\geq 1$ and $r\geq 0$.\\

	Symmetric ${\mathcal N}=2$ superpolynomials are
		labelled by ${\mathcal N}=2$ superpartitions.			The occurrence of four types of power-sums suggests  that the superpartitions are composed of four partitions. The splitting of these four types into two bosonic and two fermionic ones further entails that two of these partitions -- that associated to the product of the $\sT{p}_r$'s and that associated to the product of the  $\sD{p}_r$'s -- have distinct parts.
The superpartition $\La$ will be written as
\beq \La=( \sTD{\Lambda};\sT{\Lambda};\sD{\Lambda};\Lambda^s)\eeq
where $\sTD{\Lambda}$ and $\Lambda^s$ are usual partitions while
$\sT{\Lambda}$ and $\sD{\Lambda}$ are partitions without repeated parts.		For instance, we have for $N=2$:
		\beq p_{(;1;1;)}=\sT{p}_1\sD{p}_1=(\phi_1x_1+\phi_2x_2)(\ta_1x_1+\ta_2x_2),\qquad p_{(2;;;)}=\sTD{p}_2=\phi_1\ta_1x_1^2+\phi_2\ta_2x_2^2.\eeq
		There is a natural extension of the combinatorial scalar product defined in terms of the power-sums which preserves the dual nature of the extension of $m_\la$ and $h_\la$.\\

		However, the ultimate objective of this generalization of the theory of symmetric polynomials is to  construct the $\Nm=2$ Jack superpolynomials.
We expect those to be defined
by directly extending the  $\Nm=0,1$ definition to the $\Nm=2$ case, namely, in terms of two conditions: triangularity in the monomial basis and orthogonality.
The scalar product with respect to which we expect the yet-to-be-defined ${\mathcal N}=2$ Jack superpolynomials to be orthogonal is the $\alpha$-deformation of the power-sums scalar product just alluded to.
However, in \cite{Alarie-Vezina2015a}, we have indicated the difficulty of obtaining the Jack deformation of the classical bases along those lines. \\

Let us pinpoint the source of the problem.
We have considered in \cite{Alarie-Vezina2015a}  the characterization of the superpartitions that label the symmetric superpolynomials by three numbers: the degree of the polynomial, denoted $n$, and the number of $\phi_i$ and $\ta_j$ factors in the monomial of anticommuting variables that decorate each term in the expression of the superpolynomial in  a given sector.  Let us denote these numbers $m_\phi$ and $m_\ta$.  Now what is the problem with this description? Take the simple monomial (still for $N=2$):
\beq m_{(;1;1;)}=\phi_1x_1\ta_2x_2+\phi_2x_2\ta_1x_1\eeq
(The construction of the monomial is explained below.) Its decomposition in terms of power sums is easily found to be
\beq m_{(;1;1;)}=\sT{p}_1\sD{p}_1-\sTD{p}_2.\eeq
This preserves the sector $m_\phi=m_\ta=1$. However, it mixes for instance the sectors
corresponding to the product of $\theta_1$ and $\phi_1$ to the sector corresponding to $\theta_1 \phi_1$.  According to our earlier attempts, this mixing seems to prevent the introduction of a consistent dominance ordering, which in turn implies the impossibility of using the triangularity requirement for defining the Jack superpolynomials.\\

Heuristically, the cure for this problem is clear: the separation of the superpartition into four blocks suggests the characterization of each sector by four numbers, $n$, $\sTD{m}$, $\sT{m}$ and $\sD{m}$, the latter three
counting respectively the number of factors $\phi_i\ta_i$ (i.e., paired with the same indices), $\phi_j$ and $\ta_k$. Equivalently, $\sTD{m}$, $\sT{m}$ and $\sD{m}$ stand respectively for the length of $\sTD{\Lambda}$, $\sT{\Lambda}$ and $\sD{\Lambda}$. This refinement of the characterization of the fermionic sector is indeed a necessary requirement for the successful construction of Jack superpolynomials. But can we figure out a firmer argument for the necessity of four entries specifying a given sector?
\\

The physics of integrable $N$--body problems provide such a foundation. Recall that  the usual Jack polynomials are the eigenfunctions (with the ground-state contribution factored out)
of the Calogero-Sutherland model.  Their ${\mathcal N}=1$ extension is similarly related to the supersymmetric version of the CS model (referred to as the sCS model). We thus require the ${\mathcal N}=2$ Jack superpolynomials to be eigenfunctions of the ${\mathcal N}=2$ supersymmetric extension of the CS model (to be dubbed, for short, the s$^2$CS model), first introduced in \cite{Wyllard2000}.  This model is shown here to be integrable, as expected, by displaying four towers of $N$ bosonic mutually commuting conservation law. This naturally implies a characterization of the sectors by four quantum numbers.
\\

But this simple cure (refinement of the fermionic sector) entails the replacement of the power-sum basis by an alternative one that does not lead to sector mixing. This is one of the key technical point of our new construction and the new basis, called quasi-power-sums, is not multiplicative. The combinatorial scalar product is now defined with respect to this new basis.\\

In this way, we have succeeded in constructing the ${\mathcal N}=2$ Jack superpolynomials orthogonal with respect to this new combinatorial scalar product.  But there is more:  their construction from an appropriate symmetrization of the non-symmetric Jack polynomials, taylor made to render them s$^2$CS eigenfunctions,
implies their orthogonality with respect to an analytic scalar product.  This is compatible with their physical interpretation as wavefunctions.
\\

{The outline of the article is a follows.  In Section~\ref{sec2}, we derive the $\mathcal N=2$ supersymmetric Calogero-Sutherland model using the formalism of
\cite{Wyllard2000} and following the construction of the $\mathcal N=0,1$ cases.
In Section~\ref{sec3}, we introduce the space of $\mathcal N=2$ symmetric superfunctions and provide two simples bases: the monomial symmetric functions and the quasi-power sums.  We also present superpartitions, the combinatorial objects  which naturally index the bases, as well as the dominance ordering on
superpartitions.
In Section~\ref{sec4}, we introduce the  $\mathcal N=2$ Jack superpolynomials
from the non-symmetric Jack polynomials.  We then construct $4N$ quantities
built out of  Dunkl operators that have those polynomials as eigenfuntions,
a result that implies the integrability of the $\mathcal N=2$ supersymmetric Calogero-Sutherland model.
We then show that if a triangularity condition is imposed, it suffices to consider only 4 commuting quantities, one of them being the Hamiltonian, to characterize the  $\mathcal N=2$ Jack superpolynomials. In Section~\ref{sec5},  we 
present two scalar products, dubbed analytic and combinatorial, with respect to which the  $\mathcal N=2$ Jack superpolynomials are orthogonal.  But in order to not overburden the text, only an outline of the proofs of the orthogonality are presented.  In Section~\ref{sec6}, we give conjectures for the norm (with respect to the combinatorial scalar product) and the evaluation of the $\mathcal N=2$ Jack superpolynomials.
Finally, we discuss in Appendix~\ref{appA} how our construction of the $\mathcal N=2$ Jack superpolynomials from the
non-symmetric Jack polynomials is only of one many possible constructions.}

	\section{$\mathcal{N}=2$ supersymmetric Calogero-Sutherland  model} \label{sec2}
	% \cb{Section remani\'ee}
        Before defining the  \N{2} version of the Calogero-Sutherland model, we introduce the \N{0} and \N{1} versions.  The construction of the \N{2} version will mimic very particularly that of the \N{1} case.

        The Calogero-Sutherland (CS) model \cite{Cal,Su}
                       describes a system of $N$ identical particles (of mass $m=1$) lying on a circle of circumference $L$ and interacting pairwise:
	\begin{equation}
	H^{({\mathcal N}=0)}=\frac12\sum_{j=1}^Np_j^2+\left(\frac{\pi}{L}\right)^2\beta(\beta-1)\sum_{1\leq i<j\leq N}\frac1{\sin^2(\pi x_{ij}/L)},
	\end{equation}
	where $x_{ij}=x_i-x_j$ and $p_j=-i\d/\d x_j$ (we set $\hbar=1)$.  
       
       In  the \N{1} version of the CS model, every particle coordinate $x_j$ is paired with an anticommuting variable $\ta_j$.  In this case, the Hamiltonian is built out of two anticommuting charges $Q$ and $Q^\dagger$ (with $\ta_j^\dagger=\d/\d \ta_j$) defined in terms of a prepotential $W$ as
\beq Q=\sum_j\ta_j\big(p_j{-}i\d_{x_j}W(x)\big).
\eeq
Explicitly, the Hamiltonian is obtained as follows
\beq H^{({\mathcal N}=1)}=\frac12\{Q,Q^\dagger\},
\eeq
where $W(x)$ is determined by the requirement
\beq H^{({\mathcal N}=1)}\left|_{\substack{\\\ta_j=0}}=H^{({\mathcal N}=0)}.\right.\eeq
This fixes $W(x)$ to be
\begin{align}
				W(x) = \sum_{1\leq i<j\leq N} \frac{\beta}{2} \ln \left( \frac{1}{\sin^2 (\frac{\pi}{L}x_{ij})} \right), \			\end{align}
and the resulting Hamiltonian reads
	\begin{equation}
	H^{({\mathcal N}=1)}=\frac12\sum_{j=1}^Np_j^2+\left(\frac{\pi}{L}\right)^2\sum_{1\leq i<j\leq N}\frac{\beta(\beta-1+\ta_{ij}\ta^\dagger_{ij})}{\sin^2(\pi x_{ij}/L)},
	\end{equation}
	with $\ta_{ij}=\ta_i-\ta_j$.\\

	For the \N{2} extension, in which case $x_j$ is then paired with two independent anticommuting variables, $\ta_j$ and $\phi_j$, we need two supercharges $Q_1$ and $Q_2$ realizing the algebra
	\begin{align}\label{algebra}
				\lbrace Q_a, Q_b^\dagger \rbrace = 2 \delta_{ab} H, \qquad \lbrace Q_a, Q_b \rbrace = 0 , \qquad \lbrace Q_a^\dagger, Q_b^\dagger \rbrace = 0,\qquad a,b=1,2
			\end{align}
			where
			\beq\label{HN2}
			 H\equiv H^{({\mathcal N}=2)}.\eeq
			As noted in \cite{Wyllard2000}, the supercharges are now expressed in terms of two prepotentials: $W^{[0]}$ (the previous $W(x)$) and $W^{[1]}$:\footnote{Due to the presence of four charges, $Q_{1,2}$ and $Q_{1,2}^\dagger$, the model is said to have four supersymmetries in \cite{Wyllard2000}. Our point of view is that there are two independent charges, hence the \N{2} qualifier.}
						\begin{align}
				Q_1&=
				\sum_j \theta_j \left(
				p_j -iW_j^{[0]}(x) -i \sum_{k,l=1}^N W^{[1]}_{jkl}(x) \phi_k \phi_l^\dagger
				 \right)\\
				Q_2&=
				\sum_j \phi_j \left(
				p_j -iW_j^{[0]}(x) -i \sum_{k,l=1}^N W^{[1]}_{jkl}(x) \theta_k \theta_l^\dagger
				 \right),
			\end{align}
			where we have introduced the notation
			\beq
			W_i^{[0]} := \partial_i W^{[0]}, \quad W^{[1]}_{ijk} := \partial_i \partial_j \partial_k W^{[1]}. \eeq
It is readily seen that when the $\phi$ variables are set equal to zero, $Q_1$ reduces to $Q$ while $Q_2$ vanishes.\\

			Under the assumption that $W^{[1]} = \sum_{i<j} w(x_{ij})$, the conditions \eqref{algebra} lead to the following  Hamiltonian (we refer to \cite{Wyllard2000} for the details of this analysis)
			\begin{align}\label{HamN2}
				H_{{\text{s$^2$CS}}} = \frac{1}{2} \sum_i p_i^2 + \beta \left( \frac{\pi}{L} \right)^2 \sum_{i<j} \frac{1}{\sin^2{\frac{\pi}{L} x_{ij}}}
				\left(
					\beta - (1- \phi_{ij} \phi_{ij}^\dagger) (1 - \theta_{ij} \theta_{ij}^\dagger)
				\right).			\end{align}
				% \cb{Pas n\'ecessaire d'\'ecrire stCMS en sous-indice au vu de \eqref{HN2}. Je propose de donner le nom court s$^2$CS qui a l'avantage de se g\'en\'eraliser facilement}.
				This is thus the candidate \N{2} version of the supersymmetric CS model ($\text{s}^2\text{CS}$ for short).
				As it will be shown below, this is precisely the form of the Hamiltonian that would result from an exchange formalism projected onto the space of symmetric superfunctions.\\

			 This Hamiltonian has the same ground state as the \N{0} and \N{1} versions, namely
			 \beq \label{groundD}
                         \psi_0(x) = \Delta^\beta(x) = \prod_{j < k} \left[\sin\left(\frac{\pi x_{jk}}{L}\right)\right]^{\beta}.
			 \eeq
			 The ground-state energy  is
			 \beq E_0 = \left(\frac{\pi \beta}{L}\right)^2 \frac{N(N^2-1)}{6}.\eeq
			 Any excited-state wavefunction will be of the form $\psi(x,{\ta,\phi}) = \psi_0(x) \varphi(x,{\ta,\phi})$, with $\varphi(x,{\ta,\phi})$ a polynomial in its variables. \\

			% Examination shows us that $ (1- \phi_{ij} \phi_{ij}^\dagger) = \sT{\kappa}_{ij}$ and $ (1 - \theta_{ij} \theta_{ij}^\dagger) = \sD{\kappa}_{ij}$.
			Upon the change of variables $z_i = e^{2 \pi i x_i/L}$, the Hamiltonian becomes
			\begin{align}
				H= 2\left(\frac{\pi}{L}\right)^2
				\left(
					\sum_i (z_i \partial_i)^2 - 2\sum_{i < j} \frac{z_i z_j}{(z_{ij})^2} \beta\left( \beta - (1- \phi_{ij} \phi_{ij}^\dagger) (1 - \theta_{ij} \theta_{ij}^\dagger)\right)
					% \sT{\kappa_{ij}} \sD{\kappa_{ij}}
				\right).
			\end{align}
			It is convenient to factor out the contribution of the ground-state by a conjugation operation and perform a  rescaling to get rid of the above prefactor, defining thereby the new Hamiltonian
%			We now substract the minimal value of the spectrum of the Hamiltonian and conjugate the resulting shifted Hamiltonian with the wave function of the ground state $\psi_0(x)$
			\beq
				\mathcal{H} =  \frac{1}{2} \left( \frac{L}{\pi} \right)^2 \Delta^{-\beta} ( H - E_0) \Delta^\beta\eeq
				A simple computation yields
				\beq \mathcal{H}
				=
				\sum_i (z_i \partial_i)^2 + \beta \sum_{i<j} \frac{z_i + z_j}{z_{ij}}(z_i \partial_i - z_j\partial_j)
				-2 \beta \sum_{i<j} \frac{z_i z_j}{z^2_{ij}} (1 -(1- \phi_{ij} \phi_{ij}^\dagger) (1 - \theta_{ij} \theta_{ij}^\dagger))\label{eq.HamstCMS}.
			\eeq
			To demonstrate the integrability of this model and to study its eigenfunctions, we first need to
%			In order to present the conserved quantities of the model and obtain its eigenfunctions, we first have to
introduce the space of symmetric superfunctions.

	\section{Superpartitions and the space of symmetric superfunctions} \label{sec3}

		\subsection{Symmetric superfunctions}

			One obvious symmetry of the Hamiltonian \eqref{eq.HamstCMS} is its invariance under the simultaneous exchange of the triplet of variables, that is, under $(z_i, \theta_i, \phi_i) \leftrightarrow (z_j, \theta_j, \phi_j)$ for all $i,j$. This is the defining property of the \N{2} symmetric superfunctions. Let us define the following operators:
			\begin{align}
				K_{ij}\, : \, z_i \longleftrightarrow z_j
				, \quad
				\sT{\kappa}_{ij}\, :\, \phi_i \longleftrightarrow \phi_j
				, \quad
				\sD{\kappa}_{ij}\, :\, \theta_i \longleftrightarrow \theta_j.
			\end{align}
			The operator that produces the simultaneous exchange of the three types of variables is thus
			\begin{align}
				\K_{ij} := K_{ij} \sT{\kappa}_{ij} \sD{\kappa}_{ij} \label{def.Kscrypt}
			\end{align}
with the following action on a superfunction
			\begin{align}
				\K_{ij} f(z_i, z_j, \theta_i, \theta_j, \phi_i, \phi_j)
				= f(z_j, z_i, \theta_j, \theta_i, \phi_j, \phi_i).
			\end{align}
			Accordingly, a superfunction $f(z, \theta, \phi)$ in $N$ (triplets of) variables is said to be symmetric if and only if
			\begin{align}
				\K_{ij} f(z, \theta, \phi) = f(z, \theta, \phi)\quad \forall \quad i,j =1, \ldots, N.
			\end{align}
We will denote the space of symmetric superfunctions {in the $3N$ variables} $(x_i,\ta_i,\phi_i)$ by $\Pi^N$:
			\begin{align}
				f \in \Pi^N \iff \K_{ij} f = f \quad \forall \quad i,j=1, \ldots, N.
			\end{align}
			This space is graded by four numbers and each set of those four numbers defines a sector. To define those sectors, we must first introduce some notation.\\
			% \begin{definition}
			% 	Let
			% 	Let $f$ be a superfunction in $\Pi^N_{(n|\sTD{m},\sT{m}, \sD{m})}$, that is $f$ is a superfunction of $\Pi^N$ such that it has a total degree of $n$ in the variables $z$ and have the following property:
			% 	\begin{align}
			%
			% 	\end{align}
			% \end{definition}
			% \begin{definition}
			% 	Let $\Pi^{N}_{\sTD{m}, \sT{m}, \sD{m}}$ represent the space of symmetric superfunctions in $N$ variables and of the sector $(\sTD{m}, \sT{m}, \sD{m})$ defined such that
			% 	\begin{align}
			% 		f \in \Pi^{N}_{\sTD{m}, \sT{m}, \sD{m}} \iff \K_\omega f = f\; \forall \; \omega
			% 	\end{align}
			% \end{definition}

			We first define the fermionic sector, denoted $M$, which is itself characterized by three numbers:
			\begin{align}
				M := (\sTD{m}, \sT{m}, \sD{m}).
			\end{align}
These numbers are defined as follows: %\cy{(je propose `doublet' au lieu de `block')}
\begin{enumerate}
\item  $\sTD{m}$ is the degree of the polynomial in the {doublet} of variables $\phi_i \theta_i$;

\item  $\sT{m}$ is the degree of the polynomial in the variables $\phi_j$ that {do not form a doublet} with a variable $\theta_j$;

\item $\sD{m}$ is the degree of the polynomial in the variables $\theta_j$  that {do not form a doublet} with a variable $\phi_j$.

\end{enumerate} For example, taking $N=4$, the following superpolynomial is in the $M=(1,1,2)$ fermionic sector:
			\beq(\phi_1 \theta_1 \phi_2 + \phi_2 \theta_2 \phi_1) \theta_3 \theta_4(z_3-z_4) + \text{permutations}
			= (\phi_1 \theta_1) (\phi_2) (\theta_3 \theta_4) z_3 + \cdots
			\eeq
			Focusing on the sole term written on the right-hand side, we see that we have only one {doublet} of  variables $\phi$ and $\theta$ with the same index ( $\sTD{m}=1$), one variable $\phi$ {that do not form a doublet with} a $\theta$ of the same index ($\sT{m}=1$) and two $\theta$ variables {that do not form a doublet} with a $\phi$ of the same index ($\sD{m}=2$).
			The subspace of symmetric superfunctions (in $N$ variables) in the fermionic sector $M$ will de denoted $\Pi^N_{(M)}$. \\
			%So this superpolynomial is in the fermionic sector $M=(1,1,2)$.\\

	%		We now define this notion of sector more formally. For convenience, we
			It is convenient to introduce the following partial sums over the three numbers that define the fermionic sector
			\begin{align}
				M_1= \sTD{m}, \quad M_2 := \sTD{m} + \sT{m}, \quad M_3:= \sTD{m} + \sT{m} + \sD{m}. \label{def.MfermionPartialSum}
			\end{align}
			We then introduce the $M$-fermion monomial
			\begin{align}\label{newdef}
				[\phi; \theta]_M =
				\prod_{i=1}^{M_1}
				\phi_i \theta_i
				\prod_{j=M_1+1}^{M_2}
				\phi_j
				\prod_{k=M_2+1}^{M_3}
				\theta_k,
			\end{align}
			with the understanding that the product is 1 if
the upper bound of the product is lower than the lower bound.
			The  projector onto the monomial term $	[\phi; \theta]_M$ is
			\begin{align}
				\Pm^M =
					[\phi;\theta]_M \Big([\phi;\theta]_M\Big)^\dagger.\label{def.projectorM}
			\end{align}
% \co{
		For instance,
			\beq\Pm^{(1,1,2)}[(\phi_1 \theta_1 \phi_2 + \phi_2 \theta_2 \phi_1) \theta_3 \theta_4(z_3-z_4) + \text{permutations}]
=\phi_1 \theta_1 \phi_2  \theta_3 \theta_4(z_3-z_4).			\eeq
To recover the full symmetric superpolynomial from the projected term (e.g., the term on the right-hand side of the {previous} equality), we need to sum over the permutations of the symmetric group $S_N$ that mix the elements of the different fermionic subsectors, that is, over the elements of $S_{(M)}$ defined as
\beq \label{defGM}S_{(M)} 					:= S_N / (S_{M_1} \times S_{]M_1, M_2]}\times S_{]M_2, M_3]}\times S_{]M_3, N]})\eeq
where the following notation has been used
\beq
				S_{]j, j+k]} 	:= S_{\lbrace j+1, \cdots, j+k  \rbrace}
				\eeq
				(so that $S_N=S_{]0,N]}$).
% }
We can thus characterize a
superfunction $f$ of $\Pi^N_{(M)}$ with the condition
			\begin{align}
				\sum_{\omega \in S_{(M)}} \K_\omega \Pm^M f = f.
			\end{align}
		We finally define the subspace $\Pi^N_{(n|M)}$ as the set of polynomials $f$ in $\Pi^N$ that have degree $n$ in the variables $z$ and {that belong} to the fermionic sector $M$.

The following proposition shows that this characterization of the superpolynomials in terms of the three numbers defining the fermionic sector is sound.			\begin{proposition} \label{3nos}Let us introduce the three operators
				\begin{align}
					\sTD{F} &= \sum_{i=1}^N \phi_i \theta_i (\phi_i \theta_i)^\dagger,\qquad
					\sT{F} =  \sum_{i=1}^N \phi_i \phi_i^\dagger - \sTD{F},\qquad
					\sD{F} =  \sum_{i=1}^N \theta_i \theta_i^\dagger - \sTD{F}.
				\end{align}
				Then, for a function $f \in \Pi^N_{(M)}$, we have
				\begin{align}\label{Feig}
					\sTD{F}\, f= \sTD{m}f,\qquad
					\sT{F}\, f= \sT{m}f,\qquad
					\sD{F}\, f= \sD{m}f.
				\end{align}
			\end{proposition}
\noindent	{The proof is reported at the end of the section.}\\
	                   % \cg{J'enleverais la prochaine phrase}

%                        \cy{Il me semble que la phrase suivante est pertinente:}\co{Je suis d'accord (Luc)}

                        Since these three operators commute with the s$^2$CS Hamiltonian,
                        their eigenvalues partly characterize its eigenfunctions.

\subsection{Interlude: rederivation of the s$^2$CS Hamiltonian}
The \N{0} Hamiltonian can be recovered from the exchange formalism,
where
\begin{equation}
	H^{{\text{exch}}}=\frac12\sum_{j=1}^Np_j^2+\left(\frac{\pi}{L}\right)^2\sum_{1\leq i<j\leq N}\frac{\beta(\beta-K_{ij})}{\sin^2(\pi x_{ij}/L)},
	\end{equation}
	when the latter is restricted to the space of symmetric functions $f(x)$ such that $K_{ij}f(x)=f(x)$, i.e.,
\beq H^{({\mathcal N}=0)}=H^{\text{exch}}|_{\Pi^N}\eeq
Similarly, the \N{1} sCS Hamiltonian is recovered from
\beq H^{({\mathcal N}=1)}=H^{\text{exch}}|_{\Pi^N}\eeq
where now the restriction is on
the space of symmetric superfunctions $f(x,\ta)$ such that $K_{ij}f(x,\ta)=\kappa_{ij}f(x,\ta)$,
where \beq \label{defkappa}\kappa_{ij}=1-\ta_{ij}\ta_{ij}^\dagger\eeq
In the same way, the s$^2$CS Hamiltonian constructed previously
(cf. eq. \eqref{HamN2}) is easily recovered form
\beq \label{eqexcha}
H^{({\mathcal N}=2)}=H^{\text{exch}}|_{\Pi^N}\eeq
where in this case the restriction is on
the space of \N{2} symmetric superfunctions $f(x,\phi,\ta)$ such that $K_{ij}f(x,\phi,\ta)=\sT{\kappa}_{ij} \sD{\kappa}_{ij}f(x,\phi,\ta)$. This observation will be crucial when we study the conserved quantities of the  s$^2$CS model.

		\subsection{Superpartitions}
Bases of the space of symmetric superpolynomials, to be introduced shortly, are labeled by superpartitions \cite{Alarie-Vezina2015a}.
			A superpartition $\Lambda$ is a set of four partitions, written as,
			\beq\label{defLa} \La=(\sTDL; \sTL; \sDL; \sLs)\eeq with restrictions  on the constituent partitions:  $\sTDL$ is a standard partition in which $0$'s are 
			allowed and contribute to the length of the partition. Both $\sTL$ and $\sDL$ are partitions with distinct parts that can contain one zero which, if present, also contributes to the length. Finally, $\sLs$ is a standard partition (for which zeros are ignored).  Or, more explicitly:
                        			\begin{align}
				\sTDL_1 &\geq \sTDL_2 \geq \cdots \geq \sTDL_{\sTD{m}} \geq 0 \nonumber\\
				\sTL_1 &> \sTL_2 > \cdots > \sTL_{\sT{m}} \geq 0 \nonumber\\
				\sDL_1 &> \sDL_2 > \cdots > \sDL_{\sD{m}} \geq 0 \nonumber\\
				\sLs_1 &\geq \sLs_2 \geq \cdots \geq \sLs_{\ell ( \sLs )} > 0,
			\end{align}
where $\ell(\la)$ is the length of the partition $\la$. 			The length of the superpartition, $\ell (\Lambda)$, is the sum of the length of the constituent partitions,
			%that we write as $\ell ( \sTDL ) $, $ \ell ( \sTL ) $, $ \ell ( \sDL ) $ and $ \ell ( \sLs) $,
			\begin{align}
				\ell ( \Lambda ) = \ell ( \sTDL ) + \ell ( \sTL ) + \ell ( \sDL ) + \ell ( \sLs).
			\end{align}
			A superpartition is said to belong to the $M$-fermion sector, with $M=(\sTD{m}, \sT{m}, \sD{m})$, if
			\begin{align}
				\ell ( \sTDL ) 	= \sTD{m}, \quad
				\ell ( \sTL ) 	= \sT{m}, \quad
				\ell ( \sDL )		= \sD{m}.
			\end{align}
			With the $M_i$'s defined in \eqref{def.MfermionPartialSum}, a superpartition takes the form			\begin{align}
				\Lambda = (\Lambda_1,\dots,\Lambda_{M_1};\La_{M_1+1},\dots, \Lambda_{M_2};\Lambda_{M_2+1},\dots,\La_{M_3};\La_{M_3+1},\dots,\La_N)
			\end{align}
			with the understanding that
			\begin{align}
			  \Lambda_i &\in \sTDL \text{ for }
i \in \{1,\dots,M_1\} \\
\Lambda_i &\in \sTL \text{ for }
i \in \{ M_1+1, \dots, M_2\} \nonumber\\
				\Lambda_i &\in \sDL \text{ for } i \in \{M_2+1,\dots, M_3 \} \nonumber\\
				\Lambda_i &\in \sLs \text{ for } i \in \{M_3+1,\dots, N\}.
			\end{align}
						Finally, the bosonic degree of a superpartition is defined to be the sum of all its entries and is written $|\Lambda|$ :
			\begin{align}
				| \Lambda |  = \sum_{i=1}^{\ell (\Lambda )} \Lambda_i.
			\end{align}
Therefore,  a superpartition of bosonic degree $n$ in the $M$-fermion sector is said to be  part of the $(n|M)=(n|\sTD{m}, \sT{m}, \sD{m})$ sector.  For instance,
			\begin{align}
				\Lambda &= (3,2,2,0,0;\, 1,0;\, 3,1;\, 2,1,1)\in (16| 5, 2, 2),\nonumber\\
				\Omega &= (1,0;\, 5,2,1;\, 5,1;\, 4,1)\in (20| 2, 3,2).
			\end{align}
%			We see that $\Lambda \in (16| 5, 2, 2)$ and $\Omega $

                        \subsection{Diagrammatic representation of superpartitions} \label{subsecdiag}
		        The superpartition $\La=(\sTDL;\sTL;\sDL;\La^s)$ can also be written as a standard partition where the parts are marked according to 
                        which constituent partition they belong:
                        with overbars, underbars, both overbars and underbars, and unmarked. If there are parts which are equal, we use the ordering $\sTD{a},\sT{a},\sD{a},a$. Here is an example:
		\beq 
		\Lambda =(4,2,0;4,2,0;3,2,0;3,1)=(\sTD{4}, \sT{4}, \sD{3}, 3, \sTD{2},  \sT{2},\sD{2}, 1, \sTD{0}, \sTD{0}, \sT{0})
			\eeq
	                This	notation suggests the following  diagrammatic representation. As usual, every part is represented by a row with as many boxes as its numerical value.  If the part is marked, we add
                          a circle of a given type the end of the row:
                          a  \sTrig\ if the part is overlined, a \sCercle\  if the part is underlined and a \sDouble\  if the part is overlined and underlined.
                         We add the above ordering convention: when there are more than one circle in  a column, the ordering, from top to bottom, is \sDouble\ , \sTrig\  and \sCercle\ .
Here is the diagrammatic representation of the above  example:
		\begin{align} \label{exasuperpart}
						\Lambda =(\sTD{4}, \sT{4}, \sD{3}, 3, \sTD{2},  \sT{2},\sD{2}, 1, \sTD{0}, \sTD{0}, \sT{0}) \quad \longleftrightarrow \quad 
						\scalebox{1.5}{$\superY{
												\,&\,&\,&\,&\yTC\\
												\,&\,&\,&\,&\yT\\
												\,&\,&\,&\yC\\
												\,&\,&\,\\
												\,&\,&\yTC\\
												\,&\,&\yT\\											
												\,&\,&\yC\\
												\,\\
												\yTC\\
												\yTC\\
												\yT
												}$}
		\end{align}
		Note that there cannot be two circles of any type in the same row.

                        \subsection{Ordering on superpartitions}

		We now introduce the dominance ordering on superpartitions.
                To formulate it, we first need to a introduce a few concepts.
                For a composition $\eta \in \mathbb Z_{\geq 0}^N$, let
                $\eta^+$ be the partition obtained by reordering the entries of
                $\eta$ in weakly decreasing order.
                Considering $\Lambda$ as a composition, that is, by replacing its semicolons by commas, we define
                % \beq
                % \Lambda^{[0]} = \Lambda^+
                % \eeq
                \beq
	                \Lambda^{[0]} = \Lambda^+
                \eeq
                Also, for $1 \leq m \leq N$, let $\eta+1^m$ stand for the composition $(\eta_1+1,\dots,\eta_m+1,\eta_{m+1},\dots,\eta_N)$.  This allows to define
                % \beq
                % 	\Lambda^{[1]}= (\Lambda+1^{M_1})^+,
	               %  \qquad  \Lambda^{[2]}= (\Lambda+1^{M_2})^+ \quad {\rm and} \quad \Lambda^{[3]}= (\Lambda+1^{M_3})^+.
                % \eeq
                \beq
                	\Lambda^{[k]} = (\Lambda + 1^{M_k})^+, \quad k=1,2,3.
                \eeq
                In the diagrammatic representation defined in the previous subsection, $\Lambda^{[0]}$ correspond to the partition whose diagram is that of $\Lambda$ without its circles,  
                % $\Lambda^{[1]}$ 
                $\Lambda^{[1]}$
                correspond to the partition whose diagram is that of $\Lambda$ where every \sDouble\  is replaced by a box, 
                % $\Lambda^{[2]}$ 
                $\Lambda^{[2]}$
                correspond to the partition whose diagram is that of $\Lambda$ where every \sDouble\ or  \sTrig\  is replaced by a box, and
  				% $\Lambda^{[3]}$
  				$\Lambda^{[3]}$
  				 correspond to the partition whose diagram is that of $\Lambda$ where every circle is replaced by a box. For instance, using the example given in \eqref{exasuperpart}, we have
	% \begin{align}
	% 	&\Lambda^{[0]}=(4,4,3,3,2,2,2,1),  \qquad \qquad \qquad 
	% 	 \Lambda^{[1]}=(5,4,3,3,3,2,2,1,1,1),   \nonumber \\
	% 	&\Lambda^{[2]}=(5,5,3,3,3,3,2,1,1,1,1), \qquad \quad 
	% 	 \Lambda^{[3]}=(5,5,4,3,3,3,3,1,1,1,1)\, . 
	% \end{align}
	\begin{align}
		&\Lambda^{[0]}=(4,4,3,3,2,2,2,1),  \qquad \qquad \qquad 
		 \Lambda^{[1]}=(5,4,3,3,3,2,2,1,1,1),   \nonumber \\
		&\Lambda^{[2]}=(5,5,3,3,3,3,2,1,1,1,1), \qquad \quad 
		 \Lambda^{[3]}=(5,5,4,3,3,3,3,1,1,1,1)\, . 
	\end{align}
                % Note that it is then obvious that there is a  bijective correspondence between $(\Lambda^{[0]}, \Lambda^{[1]}, \Lambda^{[2]}, \Lambda^{[3]})$ and $\Lambda$.\\
                Note that it is then obvious that there is a  bijective correspondence between $(\Lambda^{[0]}, \Lambda^{[1]}, \Lambda^{[2]}, \Lambda^{[3]})$ and $\Lambda$.\\

The ordering on superpartitions can now be defined as
	% \begin{align}
	%   \Lambda \geq \Omega \iff  \Lambda^{[0]} \geq \Omega^{[0]},\quad
 %          \Lambda^{[1]} \geq  \Omega^{[1]},\quad
 %          \Lambda^{[2]} \geq  \Omega^{[2]} \, \,
 %                 {\rm ~and~} \, \,
 %          \Lambda^{[3]} \geq  \Omega^{[3]}\, ,
	% \end{align}
	\begin{align}
		\Lambda \geq \Omega \iff  \Lambda^{[k]} \geq \Omega^{[k]} \quad \forall\quad k=0,1,2,3.
    	% \Lambda^{[1]} \geq  \Omega^{[1]},\quad
    	% \Lambda^{[2]} \geq  \Omega^{[2]} \, \,
     %             {\rm ~and~} \, \,
    	% \Lambda^{[3]} \geq  \Omega^{[3]}\, ,
	\end{align}
   where the ordering on partitions is
   the standard dominance ordering \cite{Macdonald1998}
   			\begin{align}
				\lambda \geq \mu \quad \iff \quad  |\lambda|=|\mu| \quad {\rm and} \quad \sum_{i=1}^k \lambda_i \geq \sum_{i=1}^k \mu_i \,\quad  \forall \, k
			\end{align}
			\begin{example} Consider the three superpartitions:
				\begin{align}
					\Lambda 	&= (0,0;\, 1,0;\, 0;\, 2)\nonumber\\
					\Omega 		&= (1,0;\, 1,0;\, 0;\, 1)\nonumber\\
					\Gamma		&= (0,0;\, 2,1;\, 0;\, )
				\end{align}
				We have
				\begin{align}
					\Lambda^{[0]} 	&= (2,1)
						&	\Lambda^{[1]} &= (2,1,1,1)
						&  \Lambda^{[2]} &= (2,2,1,1,1)
						& \Lambda^{[3]} &= (2,2,1,1,1,1)\nonumber\\
					\Omega^{[0]}		&= (1,1,1)
						&  \Omega^{[1]} &= (2,1,1,1)
						&  \Omega^{[2]} &= (2,2,1,1,1)
						&  \Omega^{[3]} &= (2,2,1,1,1,1)\nonumber\\
					\Gamma^{[0]} 	&= (2,1)
						&	 \Gamma^{[1]} &= (2,1,1,1)
						&  \Gamma^{[2]} &= (3,2,1,1)
						&  \Gamma^{[3]} &= (3,2,1,1,1)
				\end{align}
				which gives
				\begin{align}
				  \Lambda^{[0]} &> \Omega^{[0]}, \quad
                                  \Lambda^{[1]}=\Omega^{[1]}, \quad
                                  \Lambda^{[2]}=\Omega^{[2]} \, \,
                                {\rm and} \,\, \Lambda^{[3]}=\Omega^{[3]}
                                 \implies \Lambda > \Omega\nonumber\\
\Gamma^{[0]} &= \Lambda^{[0]}, \quad
                                  \Gamma^{[1]}=\Lambda^{[1]}, \quad
                                  \Gamma^{[2]}>\Lambda^{[2]} \, \,
                                {\rm and} \,\, \Gamma^{[3]}>\Lambda^{[3]}
                        \implies \Gamma > \Lambda
				\end{align}
			\end{example}
		\subsection{Two bases of symmetric superpolynomials}% and  a scalar product}
                We now introduce two bases of superpolynomials that will be central in our construction of the eigenfunctions of the s$^2$CS model.
              \begin{definition}
Let  $z^\La=z_1^{\La_1} z_2^{\La_2}\cdots$.				To every $\Lambda$ in the $M$-fermion sector, we associate a monomial symmetric polynomial defined as
				\footnote{Note that this basis differs sightly from the one presented in \cite{Alarie-Vezina2015a} since we use a different ordering on $[\phi;\theta]_M$ (compare eq. (27) of \cite{Alarie-Vezina2015a} with eq. \eqref{newdef} above).
				This minor redefinition is more in line with the symmetrization of the non-symmetric Jack polynomials that plays a pivotal role in our construction.}
				\begin{align}
					m_\Lambda &= \frac{1}{f_\Lambda}\sum_{\omega \in S_N}
					\K_\omega
					 [\phi;\theta]_M z^{\Lambda}
				\end{align}
				where
				\begin{align} \label{defflam}
				    f_\Lambda &= f_{\sTDL} f_{\sLs} \quad\text{with}\qquad
				    f_\lambda = n_\lambda(0)!\, n_\lambda(1)! \cdots
				\end{align}
		In the last equation, $n_\la(i)$ stands for the multiplicity of the part $i$ in the partition $\la$ {(the part 0 being considered only for the partition $\sTDL$)}.
		 \end{definition}

			Note that $1/f_\Lambda$ is a normalization factor that can be avoided by restricting the  summation to distinct permutations.
                        \begin{proposition}[\cite{Alarie-Vezina2015a}]
                          The monomial symmetric functions $\{ m_\Lambda \}_\Lambda$
                          for all superpartitions $\Lambda$ in the $M$-fermionic sector
                          and length at most $N$ form a basis of $\Pi^N_{(M)}$.
\end{proposition}

                        \begin{example} Here are some examples of the monomial symmetric functions:
				\begin{align}
					m_{(\, ;\, 1,0;\, 2;\, )} &=
						\phi_1 \phi_2 \theta_3 (z_1 - z_2)z_3^2
						+ \phi_2\phi_3\theta_1 (z_2 - z_3)z_1^2
						+ \phi_1\phi_3\theta_2 (z_1 - z_3)z_2^2,\nonumber\\
					m_{(0,0;\, 1;\, ;\, 1)} &= \phi_1 \theta_1 \phi_2 \theta_2 \phi_3 z_3 z_4 + \text{ distinct permutations}, \nonumber\\
					m_{(2,1;\, 1,0;\, 0;\, 3,1,1 )} &=
					\phi_1 \theta_1
					\phi_2 \theta_2
					\phi_3
					\phi_4
					\theta_5
					z_1^2
					z_2^1
					z_3^1
					z_4^0
					z_5^0
					z_6^3
					z_7^1
					z_8^1
					+ \text{ distinct permutations },
				\end{align}
where in the first example we set $N=3$ while in the two other examples the number of variables in unspecified.
                              \end{example}

			{Now, a key step in the construction of the Jack superpolynomials relies on the introduction of a new basis that can be viewed as a deformation of the super power-sums introduced in \cite{Alarie-Vezina2015a}.}
			\begin{definition}
				To every $\Lambda$ in the $M$-fermion sector, we associate a symmetric function $q_\Lambda$, dubbed the quasi-power sums, defined as
				\begin{align}
					q_\Lambda = {{p}}_{\sTDL}\, m_{(;\sTL;\sDL;)} \, p_{\sLs}
				\end{align}
				where
				\begin{align}   \label{eqppp}
					{{p}}_{\sTDL} = \prod_i \left(\sum_{k=1}^N \phi_k \theta_k {z}^{\sTDL_i}_k\right)
				\end{align}
				and where $p_\la$ stands for the usual power sums:
				\beq p_\la=p_{\la_1}\cdots p_{\la_\ell}\qquad \text{with}\qquad p_n=\sum_{i=1}^N z_i^n. \eeq

			\end{definition}

			\begin{example}
				We give some examples of the $q_\Lambda$ polynomials
				\begin{align}
					q_{(\, ;\, 1,0;\, 2;\, )}	 	&= m_{(\, ;\, 1,0;\, 2;\, )}\nonumber\\
					q_{(0,0;\, 1;\, ;\, 1)} 		&= (\phi_1\theta_1 + \phi_2\theta_2 + \cdots)^2 (\phi_1z_1 + \phi_2z_2+ \cdots)(z_1+z_2+\cdots)\nonumber\\
					q_{(1,0;\, 2,1;\,1,0;\,2,2)}			&=
						(\phi_1\theta_1z_1 + \phi_2\theta_2z_2 + \cdots)
						(\phi_1\theta_1 + \phi_2\theta_2 + \cdots)
						m_{(\,;\, 2,1;\, 1,0;\, )}
						(z_1^2+ z_2^2+ \cdots)^2
						% (z_1 + z_2 + \cdots)
				\end{align}
			\end{example}
			We stress that this new basis is not multiplicative due to the non-multiplicative character of the factor $m_{(;\sTL;\sDL;)}$.

			\subsection{Proof of Proposition \ref{3nos}}
			Having introduced the monomial basis, we are now in position to prove
			 Proposition \ref{3nos}.

				\begin{proof}
				%     gg{Les deux autres preuves sont identiques.}\\
				Only the proof of the first relation in \eqref{Feig} will be presented since the other two are similar.
				The proposition is proven by direct calculation. Applying $\sTD{F}$ on an arbitrary monomial $m_\Lambda \in \Pi^N_{(M)}$ we have
				\begin{align}
					\sTD{F}\, m_\Lambda &=
					\sum_{i=1}^N
						\phi_i\theta_i(\phi_i\theta_i)^\dagger
					\sideset{}{'}\sum_{\omega \in S_N}
						\K_{\omega}
						[\phi;\theta]_M
						z^\Lambda
				\end{align}
				where the prime indicates that we  sum only over distinct permutations.
				Now, since $(\sum_i \phi_i\theta_i(\phi_i\theta_i)^\dagger)$ is $S_N$-invariant, we can move it through $\K_\omega$ to get
				\begin{align}
					\sTD{F}\, m_\Lambda &=
					\sideset{}{'}\sum_{\omega \in S_N}
						\K_{\omega}
					\sum_{i=1}^N
						\phi_{i}\theta_{i}(\phi_{i}\theta_{i})^\dagger
						[\phi;\theta]_M
						z^\Lambda \label{eq.fermion_compte_proof}
				\end{align}
				Let us now focus on the second sum:
				\begin{align}
						\sum_{i=1}^N
						\phi_{i}\theta_{i}(\phi_{i}\theta_{i})^\dagger
						[\phi;\theta]_M
						&=
						\sum_{i=1}^N
						\phi_{i}\theta_{i}(\phi_{i}\theta_{i})^\dagger
						\phi_{1}\theta_{1} \cdots \phi_{\sTD{m}}\theta_{\sTD{m}}
						\phi_{{M_1 +1}} \cdots \phi_{M_2} \theta_{M_2 +1} \cdots \theta_{M_3}
				\end{align}
				It is clear that
						$\phi_{i}\theta_{i}(\phi_{i}\theta_{i})^\dagger \phi_{j}\theta_j = \delta_{ij} \phi_j \theta_j+(1-\delta_{ij}) \phi_j \theta_j \phi_{i}\theta_{i}(\phi_{i}\theta_{i})^\dagger $.
				We can thus write				\begin{align}
						\sum_{i=1}^N
						\phi_{i}\theta_{i}(\phi_{i}\theta_{i})^\dagger
						[\phi;\theta]_M
						&= \sum_{i=1}^{\sTD{m}}
						[\phi;\theta]_M = \sTD{m} \, [\phi;\theta]_M.
				\end{align}
				Substituting  this back into \eqref{eq.fermion_compte_proof} gives us
				\begin{align}
					\sTD{F} \, m_\Lambda = \sTD{m}\, m_\Lambda
				\end{align}
				Now, since any $f \in \Pi^N_{(M)}$ has a unique decomposition on the monomial basis, we have
				\begin{align}
					\sTD{F}\, f &=
					\sTD{F} \sum_{\Lambda} c_\Lambda m_\Lambda
					=
					\sum_{\Lambda} c_\Lambda \sTD{F} m_\Lambda
					=
					\sum_{\Lambda} c_\Lambda \TDm m_\Lambda
					= \TDm\, f
				\end{align}
  and the result holds.                              
			\end{proof}

			\iffalse

			Recall the power sums introduced in \cite{Alarie-Vezina2015a} are defined (up to a possible sign) as
					\beq p_\La=\sTD{p}_{\sTD{\La}}\;\sT{p}_{\sT{\La}}\;\sD{p}_{\sD{\La}}\; p_{\La^s}\qquad\text{where}\qquad \sT{p}_\la=\prod_{i=1}^{\ell(\la)}\sT{p}_{\la_i}\qquad \text{and}\qquad \sT{p}_n=\sum_{i=1}^N \phi_i x_i^n\eeq
					and similarly for $\sD{p}_\la$ and $\sD{p}_n$ but with $\phi_i$ replaced by $\ta_i$.  In the case \N{1}, where $\La=(\La^a;\La^s)$, the quasi-power sums and the super-power sums are identical since $m_{\La^a}=\sD{p}_{\La^a}$. The necessity of working with quasi-power sums is to avoid the production of factors $\phi_i\ta_i$ from the product $\sT{p}_{\sT{\La}}\,\sD{p}_{\sD{\La}}$ since those terms would mix the different fermionic sectors.
					}\\
					\fi

	\section{Conserved quantities of the s$^2$CS model and  its eigenfunctions} \label{sec4}

		\subsection{Eigenfunctions of the s$^2$CS model in terms of the non-symmetric Jack polynomials}

		% \cb{Nouvelles subdivisions}

	%	\cb{Manque un lien d'introduction}
{
The construction of the s$^2$CS  eigenfunctions that is presented below is a direct generalization of that worked out in the CS and the  sCS models when formulated in terms of the non-symmetric Jack polynomials. We thus start with a brief review of these two known cases after summarizing the properties of the non-symmetric Jack polynomials.}\\

When discussing Jack polynomials and their generalizations, we will  comply with the standard notation and use instead the parameter $\a$ defined as
\beq
\a=1/\beta.\eeq

\subsubsection{Brief review of the non-symmetric Jack polynomials}

%		The (ordinary) Jack polynomials can be constructed from
		The non-symmetric Jack polynomials \cite{Pas,Op}, denoted  $E_\eta$, are indexed
 by a composition $\eta \in \mathbb Z_{\geq 0}^N$ and defined as follows:  $E_\eta$  is the unique polynomial of the form
		\begin{align}\label{defE}
		    E_\eta (z)= z^\eta + \sum_{\nu \prec \eta} a_{\eta, \nu} z^\nu
		\end{align}
	which simultaneously diagonalizes all Dunkl operators
		\begin{align}
		    \Dk_i E_\eta = \widehat{\eta}_i E_\eta, \label{eq.VAP_D_Eeta}
		\end{align}
		where
               the Dunkl operators $\Dk_i$ is defined as
		\begin{align}
		     \Dk_i =\a z_i \partial_i
		     	+  \sum_{j< i} \frac{z_i}{z_{ij}}(1- K_{ij})
		     	+  \sum_{j> i} \frac{z_j}{z_{ij}}(1- K_{ij}) -(i-1).
		 \end{align} 
The ordering $\prec$ in \eqref{defE} is the Bruhat order on weak compositions,
that is:
		\begin{align}
		     \nu \prec \eta \iff   \nu^+ < \eta^+ \, \, \text{ or } \, \,
		     \nu^+ = \eta^+ \text{ and } \,  w_\nu > w_\eta
                \end{align}
where $w_\eta$ is the unique permutation of minimal length such that $\eta= w_\eta \eta^+$ ($w_\eta$ permutes the entries of $\eta$) and where the Bruhat order on the symmetric group $S_N$ is such that $w_\nu > w_\eta$ iff $w_\eta$ can be written using reduced decompositions as a subword of $w_\nu$.                		The eigenvalues $\widehat{\eta}_i$ are given by
		\begin{align} \label{eigennonsym}
		    \widehat{\eta}_i = \a \eta_i -
		    \bigl(
		    	\#\lbrace
		    		j=1, \cdots, i-1 | \eta_j \geq \eta_i
		    	\rbrace
		    	+
		    	\#\lbrace
		    		j=i+1, \cdots, N | \eta_j > \eta_i
		    	\rbrace
		    \bigr).
		\end{align}
		For instance, for $\eta=(6,2,3,5,2,7,3,2)$,
                we have
                \beq
		\widehat{\eta}_1=6\a-1\quad\text{and}\quad  \widehat{\eta}_5=2\a-(4+2).
\eeq
                The action of the Dunkl operators on monomials is triangular in the Bruhat order.  To be more specific, we have
                \beq \label{triangmono}
\mathcal D_i x^\eta = \hat \eta_i \, x^\eta + \sum_{\nu \prec \eta} * \, z^\nu
                \eeq
                where the expansion coefficients are represented by $*$ for simplicity.\\

In the remainder of this subsection, we collect some relations that will later be useful,  relations that describe the action of $K_{i,i+1}$ on both the Dunkl operators and the $E_\eta$.

The Dunkl operators satisfy the degenerate Hecke relations :\\
		 \begin{gather}
		     K_{i,i+1} \Dk_{i+1} - \Dk_i K_{i, i+1} = -1, \quad \text{and }\quad
		     K_{j,j+1}\Dk_i = \Dk_i K_{j,j+1} \text{ for } i \neq j,\, j+1 .\label{eq.DegenerateHeckeRel}
		 \end{gather}
	Finally, the non-symmetric Jack polynomials have the following property
	(see for instance \cite{Knop1997})
             				\begin{align}
					K_{i,i+1} E_\eta =
						\left\lbrace \begin{array}{lr}
							\frac{1}{\delta_{i,\eta}}E_\eta + (1-\frac{1}{\delta_{i,\eta}})E_{K_{i,i+1}\eta}& \eta_i > \eta_{i+1}\\
							E_\eta& \eta_i = \eta_{i+1}\\
							\frac{1}{\delta{i,\eta}}E_\eta + E_{K_{i,i+1}\eta} & \eta_{i} < \eta_{i+1}
						\end{array}, \right. \label{eq.KiiplusEeta}
				\end{align}
			where
				\begin{align}
					\delta_{i,\eta} = \widehat{\eta}_i - \widehat{\eta}_{i+1}.
				\end{align}

\subsubsection{Construction of the Jack polynomials in terms of $E_\eta$}

It is well known that the (symmetric) Jack polynomials can be constructed out of the non-symmetric ones by a direct symmetrization process
	\beq P^{(\a)}_\la(z)=\frac1{f_\la}\sum_{\omega\in S_N} K_\omega E_{\la^R}(z)\eeq
	where $f_{\lambda}$ was defined in \eqref{defflam}
and        where $\la^R$ is the composition obtained by reordering the
entries of $\lambda$ in a weakly increasing way, that is,
given the partition $\la= (\la_1,\la_2, \dots,\la_N)$ (note that $\la$ may contain a string of zeros at the end), we have
	\beq \la^R=(\la_N,\cdots,\la_2,\la_1)\, .\eeq
%\cg{Je couperais:	The $P^{1/\beta}_\la(z)$'s are  eigenfunctions of the CS model.}\\
We stress that any composition $\eta$ that rearranges to $\lambda$
could have been used instead of $\lambda^R$.  The only difference would be that the normalization factor would not necessarily be given by  $f_\la$.

\subsubsection{Construction of the \N{1} Jack superpolynomials in terms of $E_\eta$}
		The \N{1} Jack superpolynomials can also be defined by a similar  symmetrization of the non-symmetric Jack polynomials (suitably dressed with $\ta$-terms). In this case, the superpartition $\La$ is of the form $\La=(\La^a;\La^s)$ where the parts of $\La^a$ are distinct and the partition $\La^s$ is an ordinary partition (with possibly zeros at the end).
		The Jack superpolynomial $P^{(\a)}_\La$ -- in the fermionic sector $m$ -- takes the form
	\beq P^{(\a)}_\La(z,\ta)=\frac{(-1)^{
					\binom{m}{2}} }{f_{\La^s}}\sum_{\omega\in S_N}\K_\omega\,\ta_1\cdots \ta_m\,E_{\La^R}(z)\eeq
	where $\La^R=((\La^a)^R,(\La^s)^R)$ and
        where $\K_{ij}$ stands for $\K_{ij}=K_{ij}\kappa_{ij}$, with $\kappa_{ij}$ defined in \eqref{defkappa}. It was shown that the
			$P^{(\a)}_\La(z,\ta)$'s are the eigenfunctions of the sCS model \cite{DLMsJack}.\\

\subsubsection{Construction of the \N{2} Jack superpolynomials in terms of $E_\eta$}
		By analogy,
%	The functions that diagonalize the s$^2$CS Hamiltonian and all the conserved quantities, our
the candidate  \N{2} Jack superpolynomials are constructed as follows.

%\cy{J'ai change le Claim qui suivait  en Definition}

\begin{definition}\label{cla} The \N{2} Jack superpolynomials, in the $M$-fermionic sector, are given by
% \cb{je remplace $[\phi;\theta]_{\Lambda}$ par $[\phi;\theta]_M$}
			\begin{align}
				P^{(\a)}_\Lambda (z,\phi,\ta)=
				\frac{(-1)^{
					\binom{\sT{m}}{2} + \binom{\sD{m}}{2}
					}
					}{f_{\Lambda}}\sum_{\omega \in S_N} \K_\omega [\phi;\theta]_ME_{\Lambda^R}(z) \label{eq.def-sjack}
			\end{align}
			where %$E_\eta$ is the non-symmetric Jack polynomial
			%, a quick review of these polynomial is given in \ref{Annexe.nsJack},
	%		and
			$\Lambda^R$ is the composition defined as follow
			\begin{align}
			    \Lambda^R = ((\sTDL)^R, (\sTL)^R, (\sDL)^R, (\sLs)^R).
			\end{align}
			%and $R(\lambda)$ reverses the entries of the partition $\lambda$.
\end{definition}

\subsection{Sekiguchi operators}

We will construct four families of conserved quantities in involution using
Sekiguchi operators.  The first is the usual Sekiguchi operator
\beq
S^{[0]}(u,\alpha) =
\prod_{i=1}^N(\mathcal D_i +u)
\eeq
{The other three, $S^{[k]}(u,\alpha)$, are defined as
\beq
S^{[k]}(u,\alpha) =
\sum_{M} \frac{1}{|S_{(M)}|} \sum_{\omega \in S_N}
\mathcal K_{\omega} \mathcal P^M \prod_{i=1}^{M_k}(\mathcal D_i+\alpha +u) \prod_{j=M_k+1}^{N}(\mathcal D_i +u),
\eeq
where  $k=1,2,3$.
}
% The remaining two operators, $S^{[2]}(u,\alpha)$ and  $S^{[3]}(u,\alpha)$,
% are defined similarly with $M_1$  replaced by $M_2$ and $M_3$ respectively.
\begin{proposition} \label{prop.Seki}
  We have that  $P^{(\a)}_\Lambda(z,\phi,\ta)$ given in \eqref{eq.def-sjack} is a 
  common eigenfunction of the operators $S^{[k]}(u,\alpha)$.
  % $S^{[k]}(u,\alpha)$,       $S^{[2]}(u,\alpha)$ and  $S^{[3]}(u,\alpha)$. 
  More precisely, let
  \beq
\varepsilon_\lambda(u,\alpha) = \prod_{i=1}^N (\alpha \lambda_i+1-i+u)
\eeq
where $\lambda$ is a partition (with possibly zeros at the end).
  Then
  % \beq
  % S^{[0]}(u,\alpha) P^{(\a)}_\Lambda = \varepsilon_{\Lambda^{[0]}}(u,\alpha)  P^{(\a)}_\Lambda \, , \qquad S^{[1]}(u,\alpha) P^{(\a)}_\Lambda = \varepsilon_{\Lambda^{[1]}}(u,\alpha)  P^{(\a)}_\Lambda
  %   \eeq
  \beq
  S^{[k]}(u,\alpha) P^{(\a)}_\Lambda = \varepsilon_{\Lambda^{[k]}}(u,\alpha)  P^{(\a)}_\Lambda {\rm ~~for~~} k=0,1,2,3
  % \, , \qquad S^{[1]}(u,\alpha) P^{(\a)}_\Lambda = \varepsilon_{\Lambda^{[1]}}(u,\alpha)  P^{(\a)}_\Lambda
    \eeq
% and similarly for    $S^{[2]}(u,\alpha)$ and  $S^{[3]}(u,\alpha)$.
    
			\end{proposition}
\begin{proof}
  The proof follows most of the steps of the proof of Proposition~1 of
\cite{clustering}.
It is easy to verify, using \eqref{eq.DegenerateHeckeRel},
that for all $i=1,\dots,N-1$ we have
\begin{equation}\label{eqcommu}
K_{i,i+1}(\mathcal D_i+c)(\mathcal D_{i+1}+c)=(\mathcal D_i+c)(\mathcal D_{i+1}+c)K_{i,i+1}\, ,
\end{equation}
where $c$ is an arbitrary constant.
Hence $K_{\omega} S^{[0]}(u,\a)=S^{[0]}(u,\a) K_{\omega}$
for all $\omega \in S_N$.   And since $S^{[0]}(u,\a)$ only acts on the variables $z$, we also have ${\mathcal K}_{\omega} [\phi;\theta]_M
S^{[0]}(u,\a)=S^{[0]}(u,\a) {\mathcal K}_{\omega} [\phi;\theta]_M$
for all $\omega \in S_N$.  Therefore, using \eqref{eq.def-sjack},
to prove the $S^{[0]}(u,\a)$ case
we simply need to show that
\begin{equation} \label{eqseki1}
 \left(\prod_{i=1}^N (\mathcal D_i+u) \right) E_{\Lambda^R}(z)
 = \varepsilon_{\Lambda^{[0]}}(u,\alpha) \,
E_{\Lambda^R}(z)  .
\end{equation}
Similarly, we will now show that to prove the remaining cases, it suffices
to prove that
\begin{equation} \label{eqseki2}
\left( \prod_{i=1}^{M_k} (\mathcal D_i+\alpha+u) \prod_{j=M_k+1}^N
(\mathcal D_j +u) \right)
E_{\Lambda^R}(z) = \varepsilon_{\Lambda^{[k]}}(u,\alpha) \,
E_{\Lambda^R}(z)  \, .
\end{equation}
for $k=1,2,3$.\\
% where ${\circledast}$ stands for \ding{192}, \ding{193} or \ding{194}, with
% $m=M_1,M_2$ or $M_3$ accordingly.\\

As observed before, $\mathcal P^M \mathcal K_{\omega }[\phi;\theta]_M$ is non-zero only if $\omega \in S_{(M)}$.
Using
\eqref{eq.def-sjack} again (forgetting the multiplicative factor),
we get
\begin{align}
&  S^{{[k]}}(u,\alpha)  \,  \sum_{\omega \in S_N} {\mathcal K}_\omega
[\phi;\theta]_M E_{\Lambda^R}(z) \nonumber \\
& \qquad =\frac{1}{|S_{(M)}|}
\sum_{\sigma \in S_N} {\mathcal K}_{\sigma} \left( \prod_{i=1}^{M_k} (\mathcal D_i+\alpha+u) \prod_{j=M_k+1}^N
(\mathcal D_j +u) \right) \sum_{\omega \in S_{(M)}} {\mathcal K}_\omega
[\phi;\theta]_M E_{\Lambda^R}(z)
\end{align}
{From} \eqref{eqcommu}, we can deduce that
$\prod_{i=1}^{M_k} (\mathcal D_i+\alpha+u) \prod_{j=M_k+1}^N
(\mathcal D_j +u)$ commutes with $\mathcal K_\omega [\phi;\theta]_M$
for every
$\omega \in S_{(M)}$.  Hence
\begin{equation}
   S^{{[k]}}(u,\alpha)  \,  \sum_{\omega \in S_N} {\mathcal K}_\omega
[\phi;\theta]_M E_{\Lambda^R}(z)
=
\sum_{\sigma \in S_N} {\mathcal K}_{\sigma}  [\phi;\theta]_M
\left( \prod_{i=1}^{M_k} (\mathcal D_i+\alpha+u) \prod_{j=M_k+1}^N
(\mathcal D_j +u) \right)
E_{\Lambda^R}(z)
\end{equation}
and, as claimed, \eqref{eqseki2} implies the remaining statements in the
proposition.\\

We have left to prove expressions \eqref{eqseki1} and \eqref{eqseki2}.
Let $\eta= \La^R$ and suppose that $\eta_i=r$.
It is easy to get from \eqref{eigennonsym} that the eigenvalue $\hat \eta_i$
of $\mathcal D_i$ is
\begin{equation}
\hat \eta_i = \alpha r - \#\{{\rm rows~of~} \Lambda^{[0]}
{\rm~of~size~larger~than~}r\}
 - \#\{{\rm rows~of~} \La^R {\rm~of~size~}r {\rm ~above~row~}i\} \, .
\end{equation}
Therefore, letting
\begin{equation} \label{eqji}
j_i=\#\{{\rm rows~of~} \Lambda^{[0]} {\rm~of~size~larger~than~}r\}
 + \#\{{\rm rows~of~} \La^R {\rm~of~size~}r {\rm ~above~row~}i\} +1
\end{equation}
we have $\{j_1,\dots,j_N \}=\{1,\dots,N\}$, $\Lambda_{j_i}^{[0]}=r$,
and $\hat \eta_i = \alpha \Lambda_{j_i}^{[0]} +1-{j_i}$, which gives
\eqref{eqseki1}.\\

Continuing with the same notation, we suppose that
$i$ belongs to $\{1,\dots,m\}$ and that there are $\ell$ rows of size $r$ in
$(\eta_1,\dots,\eta_m)$.  Then
$\eta_i=r$ belongs to the $\ell$ highest rows of size $r$ in $\eta$,
and thus, by \eqref{eqji}, $\Lambda_{j_i}^{[0]}$ is also one of the $\ell$ highest
rows of size
$r$ in $\Lambda^{[0]}$.
Hence, in this case
\begin{equation}
\hat \eta_i +\alpha= \alpha \Lambda_{j_i}^{[0]} +1-{j_i} +\alpha=\alpha \Lambda_{j_i}^{{[k]}} +1-{j_i} \, .
\end{equation}
If $i$ does not belong to $\{1,\dots,m\}$, then $\eta_i=r$ does not belong to the $\ell$ highest rows of size $r$ in $\eta$, and we have
\begin{equation}
\hat \eta_i = \alpha \Lambda_{j_i}^{[0]} +1-{j_i} =
\alpha \Lambda_{j_i}^{{[k]}} +1-{j_i}
\end{equation}
and \eqref{eqseki2} follows.
\end{proof}
From this proposition, we will later conclude that the Jack superpolynomials expand triangularly in the monomial basis.  But we first need to establish that the
Sekiguchi operators act triangularly on the monomial basis.
\begin{proposition} \label{propacttriang}
  We have that
  \beq
   S^{[k]}(u,\alpha) \, m_\Lambda = \varepsilon_{\Lambda^{[k]}}(u,\alpha)\, m_\Lambda + {\rm lower~terms}
   % \, , \qquad  S^{[1]}(u,\alpha) \, m_\Lambda = \varepsilon_{\Lambda^{[1]}}(u,\alpha)\, m_\Lambda + {\rm lower~terms}.
  \eeq
%   \beq
%    S^{[0]}(u,\alpha) \, m_\Lambda = \varepsilon_{\Lambda^{[0]}}(u,\alpha)\, m_\Lambda + {\rm lower~terms}\, , \qquad  S^{[1]}(u,\alpha) \, m_\Lambda = \varepsilon_{\Lambda^{[1]}}(u,\alpha)\, m_\Lambda + {\rm lower~terms}
%   \eeq
% and similarly for   $S^{[2]}(u,\alpha)$ and $S^{[3]}(u,\alpha)$.
  \end{proposition}
\begin{proof}
  % Let $m$ be a positive integer.  
  For any weak composition $\eta$, define $\eta^{[0]}$ to be $\eta^+$
  and $\eta^{[k]}$ to be $(\eta+1^{M_k})^+$.
  It was shown in \cite{BDLM} that if $z^\nu$ occurs in the expansion
  of $Y_i z^\eta$, where $Y_i$ is a Cherednik operator \cite{Cherednik1991b,Che,Mac2} (whose precise form is not needed here) \, then $\nu^{[0]} \leq \eta^{[0]}$ and $\nu^{[k]} \leq \eta^{[k]}$.  Since $\mathcal D_i$ can be obtained as a limit of $Y_i$, that is, since \cite{Kiri}
  \beq
  \mathcal D_i=  \lim_{q=t^\alpha, t\to 1}\frac{1-Y_i}{1-q} ,
  \eeq
  the result also holds for $\mathcal D_i$.  Using \eqref{triangmono}, which gives the
  triangularity in the Bruhat order, we thus have
  \beq
 \mathcal D_i z^\eta = \hat \eta_i z^\eta+ \sum_{\eta < \gamma} * \, z^\nu
  \eeq
 for certain coefficients $*$, where $\eta \leq \gamma$ iff  $\nu^{[k]} \leq \eta^{[k]}$ for $k=0,1,2,3$.
  % and $\nu^{[k]} \leq \eta^{[k]}$.
% Taking $m$ to be successively $M_1$, $M_2$ and $M_3$, we then have that
We then have that
  \beq \label{triangnew}
 \mathcal D_i z^\eta = \hat \eta_i z^\eta+ \sum_{\eta < \gamma} * \, z^\nu
  \eeq
 where the order is now
\beq
 \eta \leq \gamma
 \iff  \nu^{[k]} \leq \eta^{[k]} \text{ with } k=0,1,2,3.
 % , \,  \nu^{[1]} \leq \eta^{[1]}, \,  \nu^{[2]} \leq \eta^{[2]}\, \,  {\rm and} \,  \, \nu^{[3]} \leq \eta^{[3]}\, ,
 \eeq
 % with
 %  $\eta^{[1]}$ corresponding to $\eta^{[k]}$
 % for $m=M_1$, and similarly for  $\eta^{[2]}$ and  $\eta^{[3]}$.\\

 We finish the proof by showing that (the other cases are similar)
 \beq \label{aprouver}
S^{[1]}(u,\alpha) \, m_\Lambda = \varepsilon_{\Lambda^{[1]}}(u,\alpha)\, m_\Lambda + {\rm lower~terms}
 \eeq
It is easy to see that, up to a sign $(-1)^{\xi}$, we have
 \beq \label{eqmonosign}
 m_\Lambda= \frac{1}{f_\Lambda} \sum_{\omega \in S_N} \mathcal K_\omega [\phi;\theta]_M z^\Lambda
= (-1)^{\xi}\frac{1}{f_\Lambda} \sum_{\omega \in S_N} \mathcal K_\omega [\phi;\theta]_M z^{\Lambda^R}
 \eeq
 since the permutation $\gamma$ that sends $\Lambda$ to $\Lambda^R$ is such that
 $\mathcal K_\gamma [\phi;\theta]_M = (-1)^\xi  [\phi;\theta]_M \mathcal K_\gamma $.
 After acting with the projector $\mathcal P^M$ contained in $S^{[1]}(u,\alpha) $, we then obtain
\beq
S^{[1]}(u,\alpha) \, m_\Lambda=
\frac{(-1)^\xi}{f_\Lambda |S_{(M)}|}\sum_{\omega \in S_N}
\mathcal K_{\omega}  \prod_{i=1}^{M_1}(\mathcal D_i+\alpha +u) \prod_{j=M_1+1}^{N}(\mathcal D_i+u)   \sum_{\sigma \in S_{(M)}} \mathcal K_\sigma  [\phi;\theta]_M z^{\Lambda^R}
\eeq
Now, $ K_\sigma $ commutes with the two products of $\mathcal D_i$'s since
$\sigma \in S_{(M)}$, which gives
\beq
S^{[1]}(u,\alpha) \, m_\Lambda=
\frac{(-1)^\xi}{f_\Lambda }\sum_{\omega \in S_N}
\mathcal K_{\omega}   [\phi;\theta]_M  \prod_{i=1}^{M_1}(\mathcal D_i+\alpha +u) \prod_{j=M_1+1}^{N}(\mathcal D_i+u)   z^{\Lambda^R}
\eeq
Using \eqref{eqseki2} and \eqref{triangnew} , this gives
\begin{align}
S^{[1]}(u,\alpha) \, m_\Lambda& =
\frac{(-1)^\xi}{f_\Lambda }\sum_{\omega \in S_N}
\mathcal K_{\omega}   [\phi;\theta]_M \left( \varepsilon_{\Lambda^{[1]}}(u,\alpha)
z^{\Lambda^R} + \sum_{\eta < \Lambda^R} * \, z^\eta \right) \nonumber \\ & =
\frac{1}{f_\Lambda }\sum_{\omega \in S_N}
\mathcal K_{\omega}   [\phi;\theta]_M \left( \varepsilon_{\Lambda^{[1]}}(u,\alpha)
z^{\Lambda} + \sum_{\eta < \Lambda^R} * \, z^\eta \right)
\end{align}
Now, $\eta$ corresponds to a unique superpartition $\Omega= w_\eta \eta$, where $w_\eta \in S_{(M)}$. This correspondence is easily seen to be such that
 $\eta < \nu$ iff the corresponding superpartitions $\Omega$ and $\Gamma$
are such that $\Omega < \Gamma$.  The previous equation then immediately implies \eqref{aprouver}.
\end{proof}

\begin{proposition}  \label{proptriang} The Jack superpolynomials are unitriangularly related to
the monomials, that is,
\beq
P_\Lambda^{(\alpha)} =  m_\Lambda + \sum_{\Omega < \Lambda} d_{\Lambda \Omega}(\alpha) \, m_\Omega
\eeq
 \end{proposition}
\begin{proof}
  The proof that follows is basically the proof of the triangularity in
  Proposition 7 of \cite{BDLM}.  We include it for completeness.\\

Suppose
that there exists a term $m_\Omega$ such that $\Omega \not \leq
\Lambda$ in $P_\Lambda^{(\alpha)}$ and suppose that $\Omega$ is maximal among
those superpartitions.  Then, by Proposition~\ref{propacttriang},
the coefficient of $m_{\Omega}$ in
$S^{[k]}(u,\alpha) {  J}_{\Lambda}^\alpha$ 
% and $S^{[k]}(u,\alpha) {  J}_{\Lambda}^\alpha$
is 
equal to $d_{\Lambda \Omega} \varepsilon_{\Omega^{[k]}} (u,\alpha) $
% and $d_{\Lambda \Omega}\varepsilon_{\Omega^{[k]}} (u,\alpha)$, for ${[k]}$
% equal to  \ding{192}, \ding{193} and \ding{194}.  
for $k=0,1,2,3$.
On the other hand,
Proposition~\ref{prop.Seki} tells us that the coefficient of $m_\Omega$ in
 in
$S^{[k]}(u,\alpha) {  J}_{\Lambda}^\alpha$ 
% and $S^{[k]}(u,\alpha) {  J}_{\Lambda}^\alpha$
is
equal to $d_{\Lambda \Omega} \varepsilon_{\Lambda^{[k]}} (u,\alpha) $
% and $d_{\Lambda \Omega}\varepsilon_{\Lambda^{[k]}} (u,\alpha)$
, again for $k=0,1,2,3.$
% equal to  \ding{192}, \ding{193} and \ding{194}.
But this gives that
\beq
\varepsilon_{\Omega^{[0]}} (u,\alpha) = \varepsilon_{\Lambda^{[0]}} (u,\alpha),\,\,
 \varepsilon_{\Omega^{[1]}}(u,\alpha)= \varepsilon_{\Lambda^{[1]}}(u,\alpha), \, \,  \varepsilon_{\Omega^{[2]}}(u,\alpha)= \varepsilon_{\Lambda^{[2]}}(u,\alpha), \,\,  \varepsilon_{\Omega^{[3]}}(u,\alpha)= \varepsilon_{\Lambda^{[3]}}(u,\alpha),
\eeq
which is a
contradiction since  $\Lambda \neq \Omega$
($\varepsilon_{\Lambda^{[0]}} (u,\alpha)$,
$\varepsilon_{\Lambda^{[1]}}(u,\alpha)$, $\varepsilon_{\Lambda^{[2]}}(u,\alpha)$
and $\varepsilon_{\Lambda^{[3]}}(u,\alpha)$
uniquely determine $\Lambda$ since they uniquely determine $\Lambda^{[0]}$,
$\Lambda^{[1]}$, $\Lambda^{[2]}$ and
$\Lambda^{[3]}$).
\end{proof}
Given that the monomials form a basis of $\Pi^N_{(M)}$, we have immediately.
\begin{corollary}
   The Jack superpolynomials $\{ P_\Lambda^{(\alpha)} \}_\Lambda$
                          for all superpartitions $\Lambda$ in the $M$-fermionic sector
                          and length at most $N$ form a basis of $\Pi^N_{(M)}$.
 \end{corollary}
For $i=1,\dots,N$, let $H^{[0]}_i$ be the coefficient of $u^{N-i}$ in $S^{[0]}(u,\alpha)$.  Let
 also $H^{[1]}_i$ be the coefficient of $u^{N-i}$ in $S^{[1]}(u,\alpha)$, and similarly for $H^{[2]}_i$ and $H^{[3]}_i$.
 The previous corollary together with Proposition~\ref{prop.Seki} {imply that these $4N$ operators are in involution. We will see in the next subsection how the integrability of the 
   s$^2$CS model is then immediate since the
   s$^2$CS Hamiltonian $\mathcal H$ can be taken to be one of those
   operators.}
 \begin{corollary} \label{corcommute}
   The $4N$ quantities $H^{[0]}_i$,  $H^{[1]}_j$,
   $H^{[2]}_k$ and  $H^{[3]}_\ell$, for $i,j,k,\ell=1,\dots,N$ mutually commute when restricted to the space of symmetric superpolynomials, that is,
   \beq
   [H^{[k]}_i, H^{[l]}_j]f= 0 \quad \forall \quad k,l=0,1,2,3 \quad {\rm and} \quad i,j =1,\dots,N 
% [H^{[0]}_i, H^{{[k]}}_j]f= [ H^{{[k]}}_i, H^{{[k]}}_j]f=0
\eeq
whenever $f$ belongs to $\Pi^N_{(M)}$.
% , where $H^{{[k]}}_i$ stands for
% $H^{[1]}_i$,
%    $H^{[2]}_i$ or  $H^{[3]}_i$.
 \end{corollary}

\subsection{Complete characterization of the eigenfunctions from a minimal set of commuting operators and integrability of the  s$^2$CS model}

The Jack polynomials are fully characterized by being (1) triangular in the monomial basis, and (2) eigenfunctions of the CS Hamiltonian. Similarly, the \N{1} Jack superpolynomials are completely characterized by the triangularity condition  and that they diagonalize both the sCS Hamiltonian and another conservation law $\mathcal I$. For the \N{2} version, the condition (2) amounts to diagonalizing the s$^2$CS Hamiltonian together with three extra conservation laws (that we shall denote  $\mathcal I_{[1]}$, $\mathcal I_{[2]}$ and $\mathcal I_{[3]}$).
This implies in passing that the  s$^2$CS model is integrable. \\

We first prove that if we impose the unitriangularity, then 
the operators $H_2^{[i]}$ for $i=0,1,2,3$
% and $H_2^{\circledast}$ form$\circledast$
% equal to  \ding{192}, \ding{193} and \ding{194} 
are sufficient to characterize the Jack superpolynomials.  The eigenvalue of those operators is the coefficient of $u^{N-2}$ in $\varepsilon_{\Lambda^{[k]}}(u,\alpha)$ 
% and  $\varepsilon_{\Lambda^{[k]}}(u,\alpha)$ for
for  ${k=0,1,2,3}$.
% equal to  \ding{192}, \ding{193} and \ding{194}.  
In general, we have that
\beq
\varepsilon_\lambda^{(2)}(\alpha) :=  \varepsilon_\lambda(u,\alpha) \Big|_{u^2}= e_2(\alpha \lambda_1, \alpha \lambda_2-1,\dots,\alpha \lambda_N+1-N)
\eeq
where $e_2(x_1,x_2,x_3,\dots)=x_1 x_2 +x_1 x_3+x_2 x_3 +\cdots $ is an elementary symmetric function.\\

The following lemma will prove useful.
\begin{lemma}  \label{lemmareg} If $\mu < \lambda$ then
  $\varepsilon_\lambda^{(2)}(\alpha)\neq \varepsilon_\mu^{(2)}(\alpha)$.
  \end{lemma}
\begin{proof} It is a known easy lemma (see for instance \cite{Stan}) that we prove for completeness.
  Suppose that $\lambda$ and $\nu$ are two partitions such that $\lambda_i= \nu_i+1$ and
  $\lambda_j=\nu_j-1$ for $i<j$.  Comparing their quadratic terms in $\alpha$, we get
  \beq
  \left(\varepsilon_\nu^{(2)}(\alpha) - \varepsilon_\lambda^{(2)}(\alpha) \right)\Big|_{\alpha^2}
  =1+\nu_i-\nu_j >0 
  \eeq
since $\nu$ is a partition.
   When $\mu <\lambda$, it is well-known \cite{Macdonald1998} that one can go from $\lambda$ to $\mu$
  using steps such as those we just used to go from $\lambda$ to $\nu$.  Hence
  \beq
 \left(\varepsilon_\mu^{(2)}(\alpha) - \varepsilon_\lambda^{(2)}(\alpha) \right)\Big|_{\alpha^2}
  >0
  \eeq
  and we can conclude that $\varepsilon_\lambda^{(2)}(\alpha)\neq \varepsilon_\mu^{(2)}(\alpha)$.
\end{proof}

\begin{proposition}\label{Jaseigen}
The Jack superpolynomial $P^{(\a)}_\La$ is uniquely defined by the following two conditions:
\begin{align} &(1):\;P^{(\a)}_\La=m_\La+\text{ lower terms}\nonumber\\
  &(2):\,P^{(\a)}_\La \text{diagonalizes  the operators
    $H_2^{[0]}$,  $H_2^{[1]}$, $H_2^{[2]}$ and $H_2^{[3]}$    } \label{conditions}
\end{align}
\end{proposition}
\begin{proof}  This is again a standard proof that we include for completeness.
  Let $P^{(\a)}_\La$ and $\tilde P^{(\a)}_\La$ be two polynomials that
  obey the two conditions \eqref{conditions}.  Then
  \beq
   P^{(\a)}_\La-\tilde P^{(\a)}_\La = \sum_{\Omega < \Lambda} c_{\Lambda \Omega}(\alpha)\, m_{\Omega}
  \eeq
  for some coefficients $c_{\Lambda \Omega}(\alpha)$.  Assume that $\Gamma$ is maximal among the $\Omega$'s such that $c_{\Lambda \Omega}(\alpha)\neq 0$.
  Since $\Lambda \neq \Gamma$, we have that $\Lambda^{[k]}\neq \Gamma^{[k]}$ 
  % or $\Lambda^{[k]} \neq \Gamma^{[k]}$ for
  for either ${k}=0,1,2,3.$
  % equal to  \ding{192}, \ding{193} or \ding{194}.  
  Suppose without loss of generality that $\Lambda^{[0]} \neq \Gamma^{[0]}$.
  Applying $H_2^{[0]}$ on both sides of the previous
  equation then gives, from our assumptions,
 \beq
  \varepsilon_{\Lambda^{[0]}}^{(2)}(\alpha) \left( P^{(\a)}_\La-\tilde P^{(\a)}_\La \right)=
  \varepsilon_{\Lambda^{[0]}}^{(2)}(\alpha) \sum_{\Omega < \Lambda} c_{\Lambda \Omega}(\alpha)\, m_{\Omega}
  = H_2^{[0]} \sum_{\Omega < \Lambda} c_{\Lambda \Omega}(\alpha)\, m_{\Omega} \, ,
  \eeq
  where we used the fact, on the l.h.s., that the eigenvalue of $ P^{(\a)}_\La$
  and $\tilde P^{(\a)}_\La$
  needs to be
  $\varepsilon_{\Lambda^{[0]}}^{(2)}(\alpha)$ from Proposition~\ref{propacttriang}.
 The coefficient of $m_\Gamma$ is then, by maximality and Proposition~\ref{propacttriang}, such that
  \beq
  \varepsilon_{\Lambda^{[0]}}^{(2)}(\alpha)  \,   c_{\Lambda \Gamma} (\alpha) =
  \varepsilon_{\Gamma^{[0]}}^{(2)}(\alpha)  \,   c_{\Lambda \Gamma} (\alpha)
  \eeq
  Since $\Gamma^{[0]} < \Lambda^{[0]}$,  we have from Lemma~\ref{lemmareg} that $ \varepsilon_{\Lambda^{[0]}}^{(2)}(\alpha) \neq  \varepsilon_{\Gamma^{[0]}}^{(2)}(\alpha)$.  This leads to the contradiction that $c_{\Lambda \Gamma} (\alpha) =0$
from which we deduce that  $c_{\Lambda \Omega} (\alpha) =0$ for all $\Omega$.
  Hence
  $ P^{(\a)}_\La=\tilde P^{(\a)}_\La $ and the proof is complete.
  \end{proof}
We will now show that we can replace the four operators in the previous proposition by simpler ones.  Let
\beq
\mathcal I_{[k]} =
\sum_{M} \frac{1}{|S_{(M)}|} \sum_{\omega \in S_N}
\mathcal K_{\omega} \mathcal P^M (\mathcal D_1 +\cdots +\mathcal D_{M_k}).
\eeq
% and similarly for $\mathcal I_{[2]}$ and $\mathcal I_{[3]}$.

\begin{theorem}\label{Jeig4}
  The Jack superpolynomials $P^{(\a)}_\La$ are uniquely defined by the following two conditions:
\begin{align} &(1):\;P^{(\a)}_\La=m_\La+\text{ lower terms}\nonumber\\
&(2):\,P^{(\a)}_\La \text{diagonalizes  the operators $\mathcal H$, $\mathcal I_{[1]}$, $\mathcal I_{[2]}$ and $\mathcal I_{[3]}$}
\end{align}
where $\mathcal H$ is the Hamiltonian defined in \eqref{eq.HamstCMS}.
\end{theorem}
\begin{proof}  Let
  \beq \label{eqant}
  D:= H_1^{[0]}= \mathcal D_1+\cdots + \mathcal D_N\eeq whose eigenvalue on $P_\Lambda^{(\alpha)}$ is $\sum_i (\alpha \Lambda^{[0]}_i+1-i)=\alpha |\Lambda|-N(N-1)/2$. The eigenvalue of $D$ is thus constant on basis elements of the same degree (which is the case we are considering).  {From \eqref{eqexcha} we have that the relation
  \beq \label{eq124}
   \alpha^2 \mathcal H = D^2-2H_2^{[0]} + (N-1) D + \frac{N(N^2-1)}{6}
   \eeq
which is valid in the $\mathcal N=0$ case\footnote{This is a well know result that can be compared  for instance with (3.27) in \cite{LV} for lack of a better reference} also holds in the $\mathcal N=2$ case.}
     Hence $\mathcal H$ can replace $H_2^{[0]}$ in Proposition~\ref{Jaseigen}
  without affecting the validity of the proposition.\\

  It is also straightforward to check that
  \begin{align}
  &
    \mathcal D_1+\dots+\mathcal D_m  \nonumber \\
  &   =   \left(  \prod_{j=1}^{N}(\mathcal D_i +u) -\prod_{i=1}^{m}(\mathcal D_i+\alpha +u) \prod_{j=m+1}^{N}(\mathcal D_i+u)    \right)\Bigg |_{u^2}+m(\mathcal D_1+\dots+\mathcal D_N)+m(m-1)/2
  \end{align}
  which implies that
  \beq \label{eq126}
\mathcal I_{[k]} =H_2^{[0]} - H_2^{{[k]}} +M_kD +M_k(M_k-1)/2
  \eeq
for ${k=1,2,3}$.
Since 
% $m=M_1,M_2, M_3$ and 
the eigenvalue of $D$ are constant on basis elements of the same degree (these numbers are already encoded in $\Lambda$), we have
that $\mathcal I_{[k]}$ can also replace $H_2^{{[k]}}$ in
Proposition~\ref{Jaseigen}  without affecting the validity of the proposition.
\end{proof}
{From \eqref{eqant} and \eqref{eq124}, the  s$^2$CS Hamiltonian $\mathcal H$ can be taken to be one of the $4N$
  independent commuting quantities $H_i^{[k]}$ for $i=1,\dots,N$ and $k=0,1,2,3 $ (see Corollary~\ref{corcommute}). Hence, we have the following.
\begin{corollary} The  s$^2$CS model is integrable.
\end{corollary} } 
For completeness, we give the eigenvalues of $\mathcal H$, $\mathcal I_{[1]}$, $\mathcal I_{[2]}$ and $\mathcal I_{[3]}$.
\begin{proposition} We have that
  \beq
    \mathcal H \, P_\Lambda^{\alpha} = \epsilon_\Lambda(\alpha) \, P_\Lambda^{\alpha}
\quad {\rm and} \quad \mathcal I_{{[k]}} \, P_\Lambda^{\alpha} = \epsilon_\Lambda^{[k]} (\alpha) \, P_\Lambda^{\alpha}
    \eeq
    for ${k=1,2,3}$. %equal to \ding{192}, \ding{193} and \ding{194}.  
    The eigenvalues are given explicitly as
    \beq
    \epsilon_\Lambda(\alpha) = \sum_{i=1}^N \bigl[\alpha^2 (\Lambda^{[0]}_i)^2 +\alpha (N+1-2i) \Lambda_i^{[0]}\bigr] \quad {\rm and} \quad \epsilon_\Lambda^{[k]} (\alpha) = \sum_{i :
      \Lambda_i^{[k]}\neq \Lambda_i^{[0]} } (\alpha \Lambda^{[k]}_i + 1-i)
    \eeq
% with $ \epsilon_\Lambda^{[2]} (\alpha)$ and $ \epsilon_\Lambda^{[3]} (\alpha)$ defined similarly.
with $k=1,2,3.$
\end{proposition}
\begin{proof} Straightforward using \eqref{eq124}, \eqref{eq126}
  and the known eigenvalues of $H_2^{[0]},H_2^{[k]}$ and $D$.
\end{proof}

			\section{{Two orthogonality characterizations of the Jack superpolynomials}} \label{sec5}

		\subsection{Relation  {with}  Jack polymomials with prescribed symmetries and the analytic scalar product } %Analytical scalar product and
                The relation between the Jack superpolynomials and the non-symmetric Jack polynomials implies that we can define a scalar product with respect to which the Jack superpolynomials are orthogonal.
                We outline the details of this implication in the present section.
                
                It proves convenient to introduce as an intermediate step the relation between $P_\Lambda^{(\a)}(z,\phi,\ta)$ and the Jack polynomials with prescribed symmetry.
		\begin{definition} \label{def.PrescJack} The Jack polynomials with prescribed symmetry (the prescription being SAAS where
                  S and A stand respectively for symmetry and antisymmetry),
                  are defined as %(using the notation \eqref{compact}):
						\begin{align}
							\PrescJack_{\Lambda} =
							\frac{(-1)^{\binom{\sT{m}}{2} + \binom{\sD{m}}{2} }}{f_{\Lambda}}
								\sum_{
									{\sTD{\alpha}} ,\,
									{\sT{\alpha}} ,\,
									{\sD{\alpha}} , \,
									{\gamma}}
								(-1)^{\ell(\sT{\alpha})+\ell(\sD{\alpha})}
								\K_{\sTD{\alpha}} \K_{\sT{\alpha}} \K_{\sD{\alpha}} \K_{\gamma} E_{\Lambda^R}
						\end{align}
where $f_\Lambda$ was defined in \eqref{defflam} and where we used the compact notation
\beq\label{compact}
\sum_{{\sTD{\alpha}},\,{\sT{\alpha}},\,{\sD{\alpha}} , \,
							{\gamma}}\equiv\sum_{
						\substack{
							{\sTD{\alpha}} \in S_{M_1}\\
							{\sT{\alpha}} \in S_{]M_1, M_2]}\\
							{\sD{\alpha}} \in S_{]M_2, M_3]}\\
							{\gamma} \in S_{]M_3, N]}}
						}\eeq

					\end{definition}

\noindent                      Explicitly, the {SAAS prescription means that the } symmetrization is taken independently
                        with respect to the $\sTD{m}$ first and last $N-M_3$ variables while the antisymmetrization is taken independently
                        with respect to the variables in position $M_1+1,\dots,M_2$ and $M_2+1,\dots,M_3$.
{This entails the following corollary of Definition \ref{cla}:}
					\begin{corollary}The Jack superpolynomials \eqref{eq.def-sjack} can equivalently be written as
						\begin{align}
							P_\Lambda^{(\a)}(z,\phi,\ta) = \sum_{\omega \in S_{(M)}} \K_\omega [\phi;\theta]_M \PrescJack_{\Lambda}(z)
						\end{align}
						where we recall that
$S_{(M)}$ was defined in \eqref{defGM}.
					\end{corollary}
					\begin{proof}
						The proof is obtained by direct calculation using Definition \ref{def.PrescJack}:
						\begin{align}
							\sum_{
								\substack{
									\omega \in S_{(M)}
								}
							} \K_\omega [\phi;\theta]_M \PrescJack_{\Lambda}
				&=
							\frac{(-1)^{\binom{\sT{m}}{2} + \binom{\sD{m}}{2}}}{f_{\Lambda}}
							\sum_{\substack{\omega \in S_{(M)}\\{\sTD{\alpha}} ,\,
									{\sT{\alpha}} ,\,
									{\sD{\alpha}} , \,
									{\gamma} }}
									\K_\omega [\phi;\theta]_M
								(-1)^{\ell(\sT{\alpha})+\ell(\sD{\alpha})}
								\K_{\sTD{\alpha}} \K_{\sT{\alpha}} \K_{\sD{\alpha}} \K_{\gamma} E_{\Lambda^R}
								\\
								&=
							\frac{(-1)^{\binom{\sT{m}}{2} + \binom{\sD{m}}{2}}}{f_{\Lambda}}
							\sum_{\substack{\omega \in S_{(M)}\\{\sTD{\alpha}} ,\,
									{\sT{\alpha}} ,\,
									{\sD{\alpha}} , \,
									{\gamma} }}
									\K_\omega
								\K_{\sTD{\alpha}} \K_{\sT{\alpha}} \K_{\sD{\alpha}} \K_{\gamma}
								[\phi;\theta]_M
								 E_{\Lambda^R}
						\end{align}
						In the last equation, passing
                                                the                                            factor $[\phi;\theta]_M$ through the permutations $\K_{\sT{\alpha}},\K_{\sD{\alpha}}$ produces a sign $ (-1)^{\ell(\sT{\alpha})+\ell(\sD{\alpha})}$. Then, the combined sum over all the permutations is exactly the sum over all permutations of $S_N$.  The last line then matches the definition of $P_\Lambda^{(\a)}$ given in \eqref{eq.def-sjack}.
					\end{proof}

				%\subsection{Aand normalization}

					\begin{definition}
                                The analytic scalar product in the $M$-fermion sector is defined as
						\begin{align} 
							\prodA{ A(z, \theta, \phi) | B(z, \theta, \phi)}_\a =
							&\quad \prod_{i=1}^{N} \left(\oint \frac{dz_i}{2\pi i z_i}\right)
							\prod_{i=1}^{N} \int d\phi_i d\theta_i
							% \prod_{j=M_1+1}^{M_2} d\phi_j
							% \prod_{k=M_2+1}^{M_3} d\theta_k
               						% \Delta^{1/\a}(z)\Delta^{1/\a}(z^{-1})
               					{\prod_{k\neq l}\left(
               						1-\frac{z_k}{z_l}
               					\right)^{\frac{1}{\alpha}}
               					}
							\left[A(z, \theta, \phi)\right]^\ddagger B(z,\theta, \phi) \label{eq.definition.analytical.scal.prod}
						\end{align}
						where the $\ddagger$ operation  on the $z_i$ variables acts as $
							z_i^\ddagger = {z_i}^{-1}$ while on the anticommuting variables it is defined such that						\begin{align}
							\prod_{i \in I}\phi_i\theta_i
							\prod_{j \in J}\phi_j
							\prod_{l \in L}\theta_l
							\left(
							\prod_{i \in I}\phi_i\theta_i
							\prod_{j \in J}\phi_j
							\prod_{l \in L}\theta_l
							\right)^\ddagger
							=
							\theta_N \phi_N \cdots \theta_1 \phi_1.
						\end{align}
					\end{definition}

         The non-symmetric Jack polynomials are known to be orthogonal with respect to the scalar product \eqref{eq.definition.analytical.scal.prod} (see \cite{Op}) in the case where $M_1=M_2=M_3=0$.  It then easily follows that the Jack polynomials with prescribed symmetry are also orthogonal with respect to that scalar product.  We have indeed
                                        \begin{equation} \label{normprescribed}
\prodA{\PrescJack_{\Lambda} |\PrescJack_{\Omega} }_\alpha = \delta_{\Lambda \Omega} \, c_\Lambda(\alpha) 
                                         \end{equation} 
                                        where $c_\Lambda(\alpha)$ is a non-zero constant belonging to
                                        $\mathbb Q(\alpha)$ (see \cite{Baker2000} for an explicit formula).

					\begin{lemma} \label{lemorhto}
						The sJacks are orthogonal with respect to the scalar product \eqref{eq.definition.analytical.scal.prod} and have the following norm:
						\begin{align}
							\prodA{P_\Lambda^{(\a)} | P_\Omega^{(\a)}}_\a =
								\delta_{\Lambda \, \Omega} \, \frac{N!}{\sTD{m}!\, \sT{m}!\, \sD{m}!\, (N-M_3)!}
								\, c_\Lambda(\alpha)
						\end{align}
                                                where $c_\Lambda (\alpha)$ is the norm of the Jack polynomials with prescribed symmetry. 
					\end{lemma}
					\begin{proof}
						We know that any two symmetric superpolynomials that are not in the same sector will be de facto orthogonal. So $\Lambda$ and $\Omega$ must belong to the same sector $M$.
						%If this is so, it is clear by definition that $[\phi;\theta]_\Lambda = [\phi;\theta]_\Omega$ \cg{cette notation n'est utilisee nulle part ailleurs}.
						We thus have
						\begin{align}
							\prodA{P_\Lambda^{(\a)} | P_\Omega^{(\a)}}_\a &=
							\prodA{
								\sum_{\omega \in S_{(M)}} \K_\omega [\phi;\theta]_M\PrescJack_{\Lambda}
								\big|
								\sum_{\sigma \in S_{(M)}} \K_\sigma [\phi;\theta]_M \PrescJack_{\Omega}
							}_\a
						\end{align}
						Now, using the adjoint $ ([\phi;\theta]_M)^\dagger (\K_\sigma)^{-1}$ of $ \K_\sigma [\phi;\theta]_M $, we obtain
						\begin{align}
							\prodA{P_\Lambda^{(\a)} | P_\Omega^{(\a)}}_\a &=
							\prodA{
								\sum_{\omega, \sigma \in S_{(M)}} ([\phi;\theta]_M)^\dagger (\K_\sigma)^{-1} \K_\omega [\phi;\theta]_M \PrescJack_{\Lambda}
								|
								  \PrescJack_{\Omega}
							}_\a
						\end{align}
						Here, we see that $ ([\phi;\theta]_M)^\dagger (\K_\sigma)^{-1} \K_\omega [\phi;\theta]_M$
						 will be $0$ unless
						 $([\phi;\theta]_M)^\dagger$ and
						 $ (\K_\sigma)^{-1} \K_\omega [\phi;\theta]_M$ have the exact same fermionic content, that is, unless $\omega = \sigma$. We then get
						\begin{align}
							\prodA{P_\Lambda^{(\a)} | P_\Omega^{(\a)}}_\a &=
							\sum_{\omega \in S_{(M)}}
							\prodA{
								 \PrescJack_{\Lambda}
								|
								\PrescJack_{\Omega}
							}_\a
							\\
							&=\frac{N!}{\sTD{m}!\, \sT{m}!\, \sD{m}!\, (N-M_3)!}
								\prodA{\PrescJack_{\Lambda} |\PrescJack_{\Lambda} }_\a \delta_{\Lambda,\Omega}
						\end{align}
					\end{proof}
			Summing up, the Jack superpolynomials given in Definition \ref{cla} are equivalently characterized as follows:
\begin{theorem} The superpolynomials $\{P^{(\a)}_\La\}_\Lambda$  are defined by the two conditions:
\begin{align} &(1):\;P^{(\a)}_\La=m_\La+\text{ lower terms}\nonumber\\
&(2):\,\prodA{ P^{(\a)}_\Lambda |  P^{(\a)}_\Omega}_\alpha \propto \delta_{\Lambda \Omega}.
\end{align}
\end{theorem}
\begin{proof}  The triangularity was shown in Proposition~\ref{proptriang}.
    Given that the Gram-Schmidt process constructs a unique basis from any 
    total order compatible with the order on superpartitions, the theorem follows
    immediately. 
  \end{proof}  

		\subsection{{Combinatorial scalar product}}

We first give %in this subsection
 two simple examples of sJacks constructed from \eqref{eq.def-sjack}, or equivalently, as eigenfunctions of the s$^2$CS four basic conservation laws, expressed both in the monomial basis and the quasi-power sums. %In these expressions, we set $\alpha=1/\beta$:}
\begin{align} \label{eqexamples}
	P^{(\alpha)}_{(1;\,2, 1, 0;\,;\,)} &=
m_{(1;\,2, 1, 0;\,;\,)}
+
 {\frac {1}{\alpha + 2}}
m_{(0;\,2, 1, 0;\,;\,1)}
	\nonumber\\&=
	 \frac{1}{\left( 2+\alpha \right)}
q_{(0;\,2, 1, 0;\,;\,1)}
	- \frac{1}{\left( 2+\alpha \right)}
q_{(0;\,3, 1, 0;\,;\,)}
+	{\frac {1+\alpha}{2+\alpha}}
q_{(1;\,2, 1, 0;\,;\,)}
	\nonumber\\
	P^{(\alpha)}_{(0;\,3, 1, 0;\,;\,)} &=
m_{(0;\,3, 1, 0;\,;\,)}
+
 {\frac {1}{\alpha+1}}
m_{(0;\,2, 1, 0;\,;\,1)}
+	{\frac {1}{\alpha+1}}
m_{(1;\,2, 1, 0;\,;\,)}
	\nonumber	\\
	&=
	 \frac{1}{\left( 1+\alpha \right)}
q_{(0;\,2, 1, 0;\,;\,1)}
+	{\frac {\alpha}{1+\alpha}}
q_{(0;\,3, 1, 0;\,;\,)}
\end{align}

		We now introduce a combinatorial scalar product defined directly in terms of the quasi-power sums. It is a natural (albeit non-trivial) extension of the \N{1} scalar product for super power-sums \cite{superspace}.

\begin{definition}The scalar product is defined on the $q_\La$ basis as 				\begin{align}
					\prodC{q_\Lambda | q_\Omega}_\alpha = \alpha^{\sTD{m} + \ell(\sLs)}
					\, \xi_{\sTDL}\,  z_{\sLs} \delta_{\Lambda \Omega} \label{eq.ScalProdAlpha}
				\end{align}
				with
				\begin{align}
					z_\lambda &= \prod_{i} i^{n_{\lambda}(i)} n_\lambda(i)!\qquad \text{and}\qquad
					\xi_{\sTDL} = \prod_{i} n_{\sTDL}(i)!
				\end{align}
			\end{definition}
			%\co{Commenter sur la diff\'erence entre ceci et notre choix pr\'ec\'edent}
			%\LAV{Puisqu'il ne s'agit pas de la base des powersums tel que definit dans l'autre article, je ne crois pas qu'il y est matiere a confusion}

We can check that the $P^{(\alpha)}$'s given in \eqref{eqexamples} are orthogonal with respect to the scalar product defined in \eqref{eq.ScalProdAlpha}:
\begin{align}
	\prodC{P^{(\alpha)}_{(1;\,2, 1, 0;\,;\,)} | P^{(\alpha)}_{(0;\,3, 1, 0;\,;\,)}}_\alpha &=
	\frac{1}{(\alpha+2)(\alpha+1)} ||q_{(0;\,2, 1, 0;\,;\,1)}||^2
	- \frac{\alpha}{(\alpha+2)(\alpha+1)} ||q_{(0;\,3, 1, 0;\,;\,)}||^2
	\nonumber\\
	&=
	\frac{1}{(\alpha+2)(\alpha+1)} \alpha^2
	- \frac{\alpha}{(\alpha+2)(\alpha+1)} \alpha
	=0
\end{align}

%\co{
%We now lift this observation to an alternative definition of these Jacks based on the quasi-power sums scalar product.

\begin{claim} \label{conjecturescal} The superpolynomials $ \{P^{(\alpha)}_\La\}_\Lambda$ are defined by the two conditions:
%(1) triangularity in the monomial basis and (2) orthogonality with respect to  i.e.:
\begin{align} &(1):\; P^{(\alpha)}_\La=m_\La+\text{ lower terms}\nonumber\\
&(2):\,						\prodC{ P^{(\alpha)}_\Lambda |  P^{(\alpha)}_\Omega}_\alpha \propto \delta_{\Lambda \Omega}
\end{align}
%\begin{gather}
%						\tilde P^{\alpha}_\Lambda = m_\Lambda + \sum_{\Omega < \Lambda} a_{\Lambda \Omega} m_\Omega\\
%					\end{gather}
	where the scalar product is given in \eqref{eq.ScalProdAlpha}.
	%				are identical to the Jack superpolynomials  $P^\alpha_\La$.
\end{claim}
{In order to prove the claim, it suffices to prove that the Jack superpolynomials are orthogonal with respect to the combinatorial scalar product  \eqref{eq.ScalProdAlpha}. As we will briefly outline, the main ingredient is a symmetry property of the $4N$ commuting quantities ${H_n^{[k]}}$ with respect to a reproductive kernel of the scalar product  \eqref{eq.ScalProdAlpha}.  This symmetry property,
which is equivalent to the  self-adjointness
of the operators ${H_n^{[k]}}$ with respect to the scalar product  \eqref{eq.ScalProdAlpha}, then implies the orthogonality of the  superpolynomials $ \{P^{(\alpha)}_\La\}_\Lambda$. } This result is given as a claim since it is presented without a complete proof.  We only detail the original part of the argument, which is the formulation of the kernel.

To appropriately describe the reproductive kernel, it is convenient to define auxiliary
variables.   Let the variables $(\eta, \tilde \phi, \tilde \theta)$
 obey the relations
\newcommand{\tphi}{\tilde{\phi}}
\newcommand{\ttheta}{\tilde{\theta}}
\begin{gather}
	\{\tphi_i, \tphi_j\} =
	\{\ttheta_i, \ttheta_j\} =
	\{\tphi_i, \ttheta_j\} = 0\\
	[\eta_i,\eta_j] = 0\\
	\tphi_k \ttheta_k = 0, \quad
	\eta_i \eta_i = 0 \\
	\tphi_i \eta_i = \ttheta_i\eta_i = 0 
\end{gather}
We then introduce the space $\tilde{\Pi}_N$ of symmetric superfunctions in the $4N$ auxiliary variables $(z_i, \tphi_i, \ttheta_i, \eta_i)$, where $\tilde f$ belongs to $\tilde{\Pi}_N$ if and only if it is invariant under the simultaneous exchange of the quartet of variables $(z_i, \tphi_i, \ttheta_i, \eta_i) \longleftrightarrow (z_j, \tphi_j, \ttheta_j, \eta_j)$ for any $i,j$ $\in$ $\{1,\cdots, N\}$.  It is easy to see that if $\tilde f(z,\tphi,\ttheta,\eta)  \in  \tilde{\Pi}_N$, then we can obtain a function $f(z,\theta,\phi) \in \Pi_N$ by doing the substitution
\begin{align}
  f(z,\phi, \theta) = 
	\left[ \tilde f(z,\tphi,\ttheta,\eta) \right]_{\substack{
		\tphi_i \rightarrow \phi_i\\
		\ttheta_i \rightarrow \theta_i\\
		\eta_i \rightarrow \phi_i \theta_i
	}}
\end{align}
In particular, it is not too difficult to show that 
\begin{align}
q_\Lambda(z,\phi,\theta) =	\left[ p_\Lambda(z,\tphi,\ttheta,\eta) \right]_{\substack{
		\tphi_i \rightarrow \phi_i\\
		\ttheta_i \rightarrow \theta_i\\
		\eta_i \rightarrow \phi_i \theta_i
	}}
\end{align}
where $p_\Lambda= 	{{p}}_{\sTDL} 	{{p}}_{\sTL} 	{{p}}_{\sDL} 	{{p}}_{\Lambda^s}$
is defined in the obvious way with for instance (compare with \eqref{eqppp})
\begin{align}   
					{{p}}_{\sTDL}(z,\tphi,\ttheta,\eta) = \prod_i \left(\sum_{k=1}^N \eta_k {z}^{\sTDL_i}_k\right)
\end{align}
The auxiliary kernel $K^A(z,\tphi, \ttheta,\eta;y, \tilde{\psi},\tilde{\tau},\zeta) \equiv K^A(\tilde Z,\tilde Y;\alpha)$  is defined to be the bi-symmetric formal power series
	\begin{align}
		K^A(\tilde Z,\tilde Y;\alpha) &= 
		\prod_{ij} \frac{1}{(1-z_iy_j)^{1/\alpha}}
		\prod_{ij}(1+ \frac{\alpha^{-1}\eta_i \zeta_j}{1-z_iy_j})
		% &\quad \times \prod_{n}(1-\phi_n\theta_n \theta_n^\dagger \phi_n^\dagger)
		% \prod_{m}(1-\psi_n\tau_n \tau_n^\dagger \psi_n^\dagger)
		\prod_{ij}(1+\frac{\tphi_i\tilde{\psi}_j}{1-z_i y_j})
		\prod_{ij}(1+\frac{\ttheta_i\tilde{\tau}_j}{1-z_i y_j}).
	\end{align}
        It is then straightforward to show that
                  	\begin{align}
		K^A(\tilde Z,\tilde Y;\alpha) = 
		\sum_{\Lambda \in \text{SPar}} 
		\frac{
		1
		}{
		\alpha^{\sTD{m} + \ell(\sLs)}\xi_{\Lambda} z_\Lambda
		}
		p_{\Lambda}^\top (z,\ttheta,\tphi,\eta)
		p_{\Lambda}(y,\tilde{\tau}, \tilde{\psi}, \zeta)
	\end{align}
Taking this result from the auxiliary world back to our world, we get the following.
\begin{proposition}
  Let $ K(Z,Y;\alpha)=K(z,\phi,\ta;y,\tau,\psi;\a)$ be given by
\begin{align}  K(Z,Y;\alpha)= 	\left[ K^A(\tilde Z,\tilde Y;\alpha) \right]_{\substack{
		\tphi_i \rightarrow \phi_i, \, 	\tilde \tau_i \rightarrow \tau_i \\
		\ttheta_i \rightarrow \theta_i, \, 	\tilde \psi_i \rightarrow \psi_i \\
		\eta_i \rightarrow \phi_i \theta_i, \, 	\zeta_i \rightarrow \psi_i \tau_i 
  }}
\end{align}  
We then  have 
\begin{align}
		K(Z,Y;\alpha) = 
		\sum_{\Lambda \in \text{SPar}} 
		\frac{
		1
		}{
		\alpha^{\sTD{m} + \ell(\sLs)}\xi_{\Lambda} z_\Lambda
		}
		q_{\Lambda}^\top (z,\theta,\phi,\eta)
		q_{\Lambda}(y,{\tau},{\psi}, \zeta)
	\end{align}
which implies that $K({Z};Y;\alpha)$ is a reproducing kernel of the scalar product \eqref{eq.ScalProdAlpha}, that is,
	\begin{align}
          \prodC{K(Z,Y;\alpha)^\top | f(Z)}_{\a}= f(Y) \quad \forall \quad f \in {\Pi}_N,
	\end{align}
\end{proposition}
The main step to prove Claim~\ref{conjecturescal} is then to show that the $4N$ commuting quantities $H_n^{[k]}(Z)$ are symmetric with respect to the reproducing kernel $K(Z,Y;\alpha)$, that is
\begin{equation}
H_n^{[k]}(Z) K(Z;Y;\alpha)=H_n^{[k]}(Y) K(Z;Y;\alpha)
\end{equation}
{This can be done using  the methods described in \cite{DLMadv}.  The orthogonality of the Jack superpolynomials, and thus also the claim, then follows from standard arguments in symmetric function theory (see for instance \cite{Macdonald1998}).}

\section{Norm and evaluation} \label{sec6}

\subsection{The combinatorial norm}

In order to present our conjectured expression for the norm of the Jack superpolynomial with respect to the scalar product \eqref{eq.ScalProdAlpha}, we have to refine the description of the diagrams introduced in Subsection~\ref{subsecdiag}.
We first divide the set of boxes in a diagram into fermionic and bosonic boxes.  The fermionic boxes 
are the boxes that have a \, \sCercle\ \, both at the end of their row and at the end of their column or a\,  \sTrig\ \, both at the end of their row and at the end of their column.
The bosonic boxes are then simply the boxes that are not fermionic.
We further subdivide the set of bosonic boxes into four subsets defined as follows.
Let $B_k\Lambda$ with $k =0,1,2,3$ be the subset of bosonic boxes that are in  rows which  end with a box, a \sDouble\ , a \sTrig\ , or a \sCercle\ respectively.
We illustrate this definition with an example
in which the fermionic boxes are indicated in gray and the boxes in the  sets $B_k \Lambda$ are identified by their coordinates $(i,j)$ ($i$-th row and $j$-th column):
    
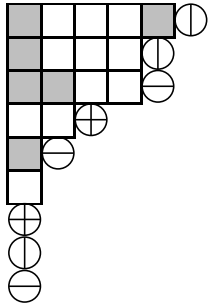
\begin{align}
    	\begin{array}{ll}
			\superYlarge{
			*(lightgray)\,&\,&\,&\,&*(lightgray)\,&\yT\\
			*(lightgray)\,&\,&\,&\,&\yT\\
			*(lightgray)\,&*(lightgray)\,&\,&\,&\yC\\
			\,&\,&\yTC\\
			*(lightgray)\,&\yC\\
			\,\\
			\yTC\\
			\yT\\
			\yC
			}
			&
			\begin{array}{l}
				B_0 \Lambda = \{ (6,1)\}\\
				B_1 \Lambda = \{ (4,1), (4,2)\}\\
				B_2 \Lambda = \{ (1,2), (1,3), (1,4), (2,2), (2,3), (2,4) \}\\
				B_3 \Lambda = \{ (3,3), (3,4) \}\\
			\end{array}
		\end{array}
    \end{align}

We are now in position to present the conjectured norm of the Jack superpolynomials. 
\begin{conjecture}
	\label{conj.norm}
The superpolynomial ${P}^{(\alpha)}_\Lambda$ has the following norm with respect to the  combinatorial scalar product 
    \eqref{eq.ScalProdAlpha}:
   %  \begin{align}\label{norm}
			%  j_\Lambda :=\prodC{{P}^{(\alpha)}_\Lambda |{P}^{(\alpha)}_\Lambda}_\a =
			% \frac{\alpha^{M_3}}{\xi_{\sTDL}} 
			% \prod_{i=0}^{3}
			% \prod_{s \in B_i \Lambda}
			% \frac{\ell_{\Lambda^{\scalebox{.5}{\circled{i \smallminus 1}}}}(s) + (a_{\Lambda^{\scalebox{.8}{\circled{0}}}}(s) + 1)\alpha}
			% {
			% 		\ell_{\Lambda^{\scalebox{.8}{\circled{i}}}}(s) + 1 + a_{\Lambda^{[3] } }(s)\alpha},
   %  \end{align}
    \begin{align}\label{norm}
			 j_\Lambda :=\prodC{{P}^{(\alpha)}_\Lambda |{P}^{(\alpha)}_\Lambda}_\a =
			\frac{\alpha^{M_3}}{\xi_{\sTDL}} 
			\prod_{i=0}^{3}
			\prod_{s \in B_i \Lambda}
			\frac{
			\ell_{\Lambda^{[i-1]}(s)} + (a_{\Lambda^{[0]}(s)} + 1)\alpha
			}{
					\ell_{\Lambda^{[i]}(s)} + 1 + a_{\Lambda^{[3]} }(s)\alpha,
			}
    \end{align}
    with the convention that $\Lambda^{[-1]} \equiv \Lambda^{[3]}$ and where, for a partition $\la$ and its conjugate $\la'$ (obtained by interchanging rows and columns), we have that the arm-length and the leg-length are respectively given by
        \beq a_\la(s)=\la_i-j\quad\text{and}\quad \ell_\la(s)=\la_j'-i,\qquad \text{where}\; s=(i,j).\eeq
    %where $\ell_\lambda(s)$ is the "leg" of the box $s$ and correspond to the number of boxes that are directly under the box $s$ in the diagram of the partition $\lambda$, likewise $a_\lambda(s)$ is the "arm" of the box $s$ and correspond to the number of boxes on the right of the box $s$. 
\end{conjecture}
This conjecture has been tested for every superpartition of the following sectors:
\begin{align}\label{tests}
&(1|1,1,1), \quad
(2|1,1,1), \quad
(2|2,1,1), \quad
(3|2,1,0), \quad
(3|2,1,1), \quad
(3|2,1,2), \nonumber\\&
(3|2,2,0), \quad
(3|2,2,1), \quad
(3|2,2,2), \quad
(4|2,2,2), \quad
(5|3,2,2), \quad
(6|2,2,2).\end{align} These sectors all together contain 171 superpartitions. 
As further evidence of the validity of the conjecture, we stress that the norm has the correct reduction 
for  \N{1} superpartitions.  Let us make this explicit.  In the case where there is only one type of circles, we need to replace $\La^{(3)}$ by $\La^{(1)}$ (so that now $\La^{(-1)}=\La^{(1)}$) and to restrict the product to $i=0,1$. Thus, the \N{1} special case of \eqref{norm} reads (with $M_3=m$)
    \begin{align}\label{norm}
			 j_\Lambda \overset{{\mathcal N}=1}=&
			 %	\a^m			\prod_{i=0}^{1}
		%	\prod_{s \in B_0 \Lambda}
		%	\frac{\ell_{\Lambda^{\scalebox{.75}{\circled{i-1}}}}(s) + (a_{\Lambda^{\circled{0}}}(s) + 1)\alpha}{
		%			\ell_{\Lambda^{\circled{i}}}(s) + 1 + a_{\Lambda^{\circled{1}}}(s)\alpha},
%\nonumber\\
\a^m   \prod_{s \in B_0 \Lambda}
			\frac{\ell_{\Lambda^{[1]}}(s) + (a_{\Lambda^{[0]}}(s) + 1)\alpha}{
					\ell_{\Lambda^{[0]}}(s) + 1 + a_{\Lambda^{[1]}}(s)\alpha}
					\prod_{s \in B_1 \Lambda}
			\frac{\ell_{\Lambda^{[0]}}(s) + (a_{\Lambda^{[0]}}(s) + 1)\alpha}{
					\ell_{\Lambda^{[1]}}(s) + 1 + a_{\Lambda^{[1]}}(s)\alpha}
			.		\end{align}
			Note that the boxes in $B_1\La$ belong to rows that end with a circle and because they are bosonic they cannot have a circle in their column. Therefore, for the boxes in $B_1\La$, we have that $\ell_{\Lambda^{[0]}}(s)=\ell_{\Lambda^{[1]}}(s)$ so that we can rewrite $j_\La$ under the compact form:
			\beq  j_\Lambda \overset{{\mathcal N}=1}=\a^m   \prod_{s \in B_0\La\cup B_1 \Lambda}
			\frac{\ell_{\Lambda^{{[1]}}}(s) + (a_{\Lambda^{[0]}}(s) + 1)\alpha}{
					\ell_{\Lambda^{[0]}}(s) + 1 + a_{\Lambda^{[1]}}(s)\alpha}
					\eeq
					which is precisely the formula given in \cite[Eq. (18)]{desrosiers2012}.
					 \\

Back to the general \N{2} case,  we see that each bosonic box contributes to a factor in the conjectural expression for the norm.
Here is an example,
%We  illustrate the formula with the following example 
where the contribution of each bosonic box  is written within the box:% it corresponds to. 
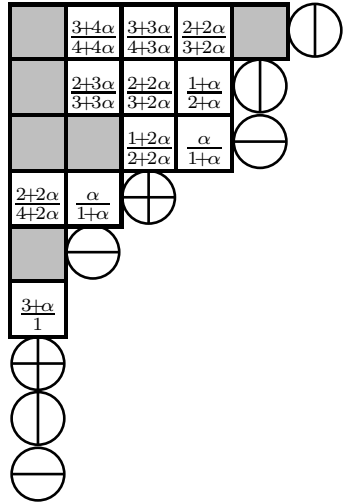
\newcommand{\tabfrac}[2]{\scalebox{.4}{$\frac{#1}{#2}$}}
	\begin{align*}
    	\begin{array}{ll}
			\superYBig{
			*(lightgray)\,&\scalebox{.4}{$\frac{3\!+\!4\alpha}{4\!+\!4\alpha}$}&\scalebox{.4}{$\frac{3\!+\!3\alpha}{4\!+\!3\alpha}$}&\tabfrac{2\!+\!2\alpha}{3\!+\!2\alpha}&*(lightgray)\,&\yT\\
			*(lightgray)\,&\tabfrac{2\!+\!3\alpha}{3\!+\!3\alpha}&\tabfrac{2\!+\!2\alpha}{3\!+\!2\alpha} &\tabfrac{1\!+\!\alpha}{2\!+\!\alpha} &\yT\\
			*(lightgray)\,&*(lightgray)\,&\tabfrac{1\!+\!2\alpha}{2\!+\!2\alpha}&\tabfrac{\alpha}{1\!+\!\alpha}&\yC\\
			\tabfrac{2\!+\!2\alpha}{4\!+\!2\alpha} &\tabfrac{\alpha}{1\!+\!\alpha} &\yTC\\
			*(lightgray)\,&\yC\\
			\tabfrac{3 \!+\! \alpha}{1} \\
			\yTC\\
			\yT\\
			\yC
			}
		\end{array}
    \end{align*}
	The norm is the product of all theses contributions times the prefactor $\a^{M_3}/\xi_{\sTDL} =\a^8/1$, which gives			%\begin{array}{l}
			%	\alpha^{M_3} = \alpha^{8}\\
			%	\xi_{\sTDL} = 1\\
			\beq	j_\Lambda = \alpha^{8}\left[ \frac{\alpha^2 (3+ \alpha) (1 + 2\alpha) (2+ 3\alpha) (3 + 4\alpha)}{2 (2 + \alpha)^2 (3 + 2\alpha)^2 (4 + 3\alpha)}\right].
\eeq%			\end{array}

		\subsection{Evaluation}
                The evaluation of the Jack polynomials $J_\lambda^{(\alpha)}$ refer to its explicit expression (in terms of $\lambda$, $\alpha$ and the number of variables $N$)
                  when all variables $x_i$ are set equal to 1. In the \N{1} case, because there is a part of the superpolynomial that is antisymmetric in the $x_i$'s, setting $x_i=1$ for all $i$ makes the polynomial vanish if its fermionic sector $m$ is greater than 1.  Note that the fermionic variables $\ta_i$ are not set to a definite value. The proper way to do the evaluation is by: 
\begin{enumerate}
\item  removing the monomial prefactor $\ta_1\cdots \ta_m$, 
\item dividing the result by the Vandermonde determinant in the variables $x_1,\cdots x_m$, and 
\item setting all the variables $x_i=1$.
\end{enumerate}
		In the resulting combinatorial expression for the evaluation, the contributing boxes are those of the skew diagram $\La^{[1]}/\delta_m$ where $\delta_m$ is the fermionic core, defined as $\delta_m=(m,m-1,\cdots,1)$.\\

		The procedure for the evaluation of the \N{2} superpolynomial is  a direct generalization of the \N{1} case. 
%We now define the set of coordinates $\Delta \Lambda$ on a superpartition $\Lambda$. Given a superpartition of the $M$-fermionic sector, let 
We first introduce the fermionic core\begin{align}
	\delta_{\sT{m}, \sD{m}} = (\sT{m}, \sT{m}-1, \cdots, 1) \cup (\sD{m}, \sD{m}-1, \cdots, 1)
\end{align}
and then define the skew diagram 
\beq \Delta \Lambda=\La^{[3]}/\delta_{\sT{m}, \sD{m}},\eeq
 i.e.,  the set of boxes of $\Lambda^{[3]}$ that are not in $\delta_{\sT{m}, \sD{m}}$. 
Here is an example
\newcommand{\yG}{*(lightgray)}
\begin{align}
	\Lambda = \superY{
	\,&\,&\,&\,&\,&\yT\\
	\,&\,&\,&\,&\yT\\
	\,&\,&\,&\,&\yC\\
	\,&\,&\,&\yTC\\
	\,&\,&\,&\yC\\
	\,&\,&\yT\\
	\,&\,\\
	\,&\yT\\
	\yC
	}
	\begin{array}{l}
		\delta_{4,3} = 
		\superY{
		\yG&\yG&\yG&\yG\\
		\yG&\yG&\yG\\
		\yG&\yG&\yG\\
		\yG&\yG\\
		\yG&\yG\\
		\yG\\
		\yG
		}
	\end{array}
	\begin{array}{ll}
	\Delta \Lambda = \superY{
	\yG&\yG&\yG&\yG&\,&\,\\
	\yG&\yG&\yG&\,&\,\\
	\yG&\yG&\yG&\,&\,\\
	\yG&\yG&\,&\,\\
	\yG&\yG&\,&\,\\
	\yG&\,&\,\\
	\yG&\,\\
	\,&\,\\
	\,
	}
	\end{array}
\end{align}
with the understanding that $\Delta\La$ is the set of white boxes in the last diagram.\\
%\newcommand{\yG}{*(lightgray)}
%We represent the boxes that belong to $\Delta \Lambda$ as white boxes and those that do not by the gray boxes:

For a superpolynomial  $F(x, \theta, \phi)$ in the $M$-fermionic sector, with $N\geq M_3$, we define its evaluation as
	\begin{align}
		E_{N, M}\left[ F(x,\theta,\phi) \right] := 
		\left[ 
			\frac{[\phi;\theta]_M^\dagger F(x,\theta,\phi)}{V_M(x)}
		\right]_{x_1 = x_2 = \cdots = x_N = 1},
	\end{align}
	where 
	\begin{align}
		V_M(x) = \prod_{M_1 < i < j \leq M_2}(x_i - x_j) \prod_{M_2 < k < l \leq M_3}(x_k - x_l),
	\end{align}
	is a product of Vandermonde determinants in the variables $x_{M_1+1}, ..., x_{M_2}$ and $ x_{M_2+1}, ..., x_{M_3}$ respectively. 
We are now in position to formulate our conjectural expression for the evaluation.

\begin{conjecture} The evaluation of the Jack superpolynomial ${P}^{(\alpha)}_\Lambda$, for  $N \geq \ell(\Lambda)$, is 
	\begin{align}
		E_{N,M}\left[ {P}^{(\alpha)}_\Lambda \right] = 
		\binom{N-\sT{m} -\sD{m}}{\sTD{m}}^{-1} 
		\frac{
			\prod_{s \in \Delta \Lambda}
			\bigl(N - \ell^\prime_{\Lambda^{[3]}}(s) + \alpha\, a_{\Lambda^{[3]}}^\prime(s)\bigr)
			} 
			{
			\xi_{\sTDL}
			\prod_{i=0}^{3}
			\prod_{s \in B_i \Lambda}
			\bigl(
			{\ell_{\Lambda^{[i]}}(s) + 1 + \a a_{\Lambda^{[3]}}(s)}
			\bigr)
			},
	\end{align}
where, for a partition $\la$, 
	\beq a'_\la(s)=j-1\quad \text{and}\quad \ell'_\la(s)=i-1,\qquad \text{with}\quad s=(i,j).
	\eeq
	%$\ell^\prime_\lambda(s)$ and $a^\prime_\lambda(,s)$ are respectively the "co-leg" and "co-arm" which are respectively defined as the number of boxes above and to the left of the box $s$. 
\end{conjecture}

This conjecture has been tested for $N=3,4,5,6$ in each of the following sectors (with the understanding that the evaluation only makes sense whenever $N\geq \ell(\Lambda)$):
$(1|1,0,1)$, $(2|1,0,1)$, $(3|2,0,1)$, $(3|2,2,1)$, $(3|2,2,2)$, 
$(n|0,0,\sTD{m})$ with $n=1\ldots4$ and $\sTD{m}=1\ldots3$. It was also tested for $(4|0,0,3)$ with $N=7$. These sectors represent together 120 superpartitions. Note also that this formula is \emph{de facto} correct for any sector $(n|\sT{m}, 0,0)$ or $(n|0, \sD{m},0)$ since, in these cases, it reduces to the \N{1} evaluation formula presented in \cite{desrosiers2012}.\\

We illustrate the formula for the Jack superpolynomial in the fermionic sector $M=(1,2,2)$  indexed by the superpartition
\begin{align}
	\Lambda =
	\superY{
	\,&\,&\,\\
	\,&\,&\yTC\\
	\,&\,&\yT\\
	\,&\yC\\
	\yT\\
	\yC
	}
\end{align}
We will consider the case  $N=6=\ell(\La)$. First, we have $\xi_{\sTDL}=1$ and the binomial coefficient is %\begin{align}
$	\binom{N-\sT{m}-\sD{m}}{\sTD{m}} = 2.$
%\end{align}
We next compute the product in the numerator. Here $\delta_{2,2} = (2,2,1,1)$. The contribution of each box that belongs to $\Delta \Lambda$ is given explicitly in the following diagram
\begin{align}
	\superYbig{
	\yG&\yG&\tabsmall{6\!+\!2\alpha}\\
	\yG&\yG&\tabsmall{5\!+\!2\alpha}\\
	\yG&\tabsmall{4\!+\!\alpha}&\tabsmall{4\!+\!2\alpha}\\
	\yG&\tabsmall{3\!+\!\alpha}\\
	\tabsmall{2}\\
	\tabsmall{1}
	} \;\implies\; \prod_{s \in \Delta \Lambda}
			(N - \ell^\prime_{\Lambda^{[3]}}(s) + \alpha\, a_{\Lambda^{[3]}}^\prime(s))
			= 8 (2+\alpha)(3+\alpha)^2(4+\alpha)(5+2\alpha).
\end{align}
%This yields
%\begin{align}
%			\prod_{s \in \Delta \Lambda}
%			(N - \ell^\prime_{\Lambda^{\circled{3}}}(s) + \alpha\, a_{\Lambda^{\circled{3}}}^\prime(s))
%			= 8 (2+\alpha)(3+\alpha)^2(4+\alpha)(5+2\alpha)
%\end{align}
For the denominator, we have
\begin{align}
	\superYbig{
	\tabsmall{4\!+\!2\alpha}&\tabsmall{3\!+\!\alpha}&\tabsmall{1}\\
	\tabsmall{3\!+\!2\alpha}&\tabsmall{2\!+\!\alpha}&\yTC\\
	\yG&\tabsmall{1\!+\!\alpha}&\yT\\
	\yG&\yC\\
	\yT\\
	\yC
	} \;\implies\;\prod_{i=0}^{3}
	\prod_{s \in B_i \Lambda}
	(
		{\ell_{\Lambda^{[i]}}(s) + 1 + \a a_{\Lambda^{[3]}}(s)}	)
	= 2 (1+\alpha)(2+\alpha)^2 (3+\alpha)(3+2\alpha).
\end{align}
%Collecting these contributions gives
%\begin{align}
%	\prod_{i=0}^{3}
%	\prod_{s \in B_i \Lambda}
%	(
%		{\ell_{\Lambda^{\circled{i}}}(s) + 1 + a_{\Lambda^{\circled{3}}}(s)\alpha}	)
%	= 2 (1+\alpha)(2+\alpha)^2 (3+\alpha)(3+2\alpha)
%\end{align}
Collecting these contributions gives %\cy{Pourquoi le tilde sur ${P}^{(\alpha)}$?}
\begin{align}
	E_{6,(1,2,2)}[ {P}^{(\alpha)}_{\superYsmall{
		\,&\,&\,\\
		\,&\,&\yTC\\
		\,&\,&\yT\\
		\,&\yC\\
		\yT\\
		\yC	
	}} ]= \frac{2(3+\alpha)(4+\alpha)(5+2\alpha)}{(1+\alpha)(2+\alpha)(3+2\alpha)}.
\end{align}
\appendix

\section{Different prescribed symmetries} \label{appA}

The construction of the Jack superpolynomials in terms of the Jack polynomials with prescribed symmetry was discussed in Section 5.1.  The prescribed symmetry underlying our construction is of type SAAS -- cf.  Definition \ref{def.PrescJack}.  Although this ordering of the (anti)-symmetrization operation appears to be rather natural (up to the trivial permutation of $A_\phi$ and $A_\theta$ which amounts to  a simple relabeling of variables), one can ask wether this is the only possibility  for defining the \N{2} version of the Jacks.\\

For stability reasons when the number of variables is set to infinity, it is natural to let the symmetrization associated to the unmarked entries in a superpartition be to the right. Given the aforementioned trivial permutation between $A_\phi$ and $A_\theta$ we actually only have to consider the alternatives ASA or AAS for the first three constituent partitions with  the understanding that
\begin{align}
		S_{\phi\theta}A_\phi A_\theta 
			&\longrightarrow 
			\left\{
			\begin{array}{l}
				M = (\sTD{m}, \sT{m}, \sD{m})\\
				\Lambda = (\sTDL;\sTL;\sDL;\sLs)
			\end{array}	
			\right.
			\\
		A_\phi S_{\phi\theta} A_\theta
			&\longrightarrow
			\left\{
			\begin{array}{l}
				M = (\sT{m}, \sTD{m}, \sD{m})\\
				\Lambda = (\sTL;\sTDL;\sDL;\sLs)
			\end{array}	
			\right.
			\\
		A_\phi A_\theta S_{\phi\theta}
			&\longrightarrow
			\left\{
			\begin{array}{l}
				M = (\sT{m}, \sD{m}, \sTD{m})\\
				\Lambda = (\sTL;\sDL;\sTDL;\sLs)
			\end{array}	
			\right.
\end{align}
with  the partial sums $M_i$ being changed accordingly.  To each case, there corresponds a  specific dominance ordering.  Accordingly, the ordering of symbols in the diagrammatic representation of superpartitions must be coherent with the choice of symmetrization, that is\\
 	\begin{align}
 		\begin{array}{lll}
 			\text{SAA: }
 			\superY{
 			\yTC\\
 			\yT\\
 			\yC
 			},
 			&
 			\text{ASA: }
 			\superY{
 			\yT\\
 			\yTC\\
 			\yC
 			},
 			&
 			\text{AAS: }
 			\superY{
 			\yT\\
 			\yC\\
 			\yTC
 			}
 		\end{array}
 	\end{align}
\\

Let us introduce a compact notation to cover these alternative constructions. Let
\begin{align}
	\Lambda^{(1)} = \sTDL,\quad 
	\Lambda^{(2)} = \sTL, \quad
	\Lambda^{(3)} = \sDL, \quad
	\Lambda^{(4)} = \sLs
\end{align}
and $\sigma$ be an element of $S_3$, so that
%\begin{gather}
	%\sigma \in S_3, \quad
	\beq \sigma(\Lambda) = (\Lambda^{(\sigma(1))}; \Lambda^{(\sigma(2))}; \Lambda^{(\sigma(3))}; \Lambda^{(4)})\eeq
	The reverse superpartition would now read
	\beq
	(\sigma(\Lambda))^R = ((\Lambda^{(\sigma(1))})^R; (\Lambda^{(\sigma(2))})^R; (\Lambda^{(\sigma(3))})^R; (\Lambda^{(4)})^R)\eeq
	Finally, with
	\beq
	\Xi^{(i)}_j =
	\left\{
	\begin{array}{ll}
		\phi_j\theta_j &\text{for } i =1\\
		\phi_j &\text{for } i =2\\
		\theta_j &\text{for } i =3
	\end{array}
	\right.\eeq
	the modified version of 
	$[\phi;\theta]_M$ is
	\beq	[\phi;\theta]_M^{\sigma} = 
	\left[\prod_{i=1}^{M_1} \Xi^{\sigma(1)} \right]
	\left[\prod_{j=M_1+1}^{M_2} \Xi^{\sigma(2)} \right]
	\left[\prod_{k=M_2+1}^{M_3} \Xi^{\sigma(3)} \right]
\eeq
The polynomial with $\sigma(\text{SAA})$ symmetry and labelled by $\sigma(\La)$ would the be  given by
\begin{align}
	P^{(\a)}_{\sigma(\Lambda)} = \frac{(-1)^{ \binom{2}{\sT{m}} + \binom{2}{\sD{m}} }} {f_\Lambda} 
	\sum_{\omega \in S_N} \K_\omega [\phi;\theta]^\sigma_M E_{(\sigma(\Lambda))^R}
\end{align}
Are these proper candidates for \N{2} versions of the Jack superpolynomials? 
Remarkably, it seems that it is indeed the case given the following properties/conjectures:

\begin{enumerate}
 	\item These polynomials are still eigenfunction of the $\text{s}^2$CMS Hamiltonian.
 	\item  By construction, they are still  orthogonal with respect to the analytic scalar product \eqref{eq.definition.analytical.scal.prod}.
	
	\item They appear to be also orthogonal with respect to the combinatorial scalar product \eqref{eq.ScalProdAlpha}.
	
\item The expression for the norm given in Conjecture \ref{conj.norm} is still valid if the role of the $B_i \Lambda$'s is changed according to the permutation of the constituent partitions $\sTDL, \sTL$ and $\sDL$. (This version of the conjecture has been tested for every permutation of (SAA) for all the cases listed in eq. \eqref{tests}.) \end{enumerate}

\begin{table}[ht]
\caption{Sample of small degree Jack superpolynomials  expanded in the monomial basis.}
% \coo{p-e assez d'ex comme ca}
\label{tab2}
\small{ %\footnotesize{
\begin{center}
\begin{tabular}{|c|l|}
\hline
& \\
Sector($\La$) & $P_\La^{\a} $  \\ & \\ \hline & \\
$(1|1,1,1)$ % & $P_{(0;1)}=m_{(0;1)}$ \\ %& \\
  &   $P^{\a}_{(0;\,0;\,0;\,1)} =
	m_{(0;\,0;\,0;\,1)}$ \\ & \\
  & 	$P^{\a}_{(0;\,0;\,1;\,)} =
m_{(0;\,0;\,1;\,)}
- \frac{1}{\left( \alpha+3 \right)}
m_{(0;\,0;\,0;\,1)}
$
\\ &\\ & $P^{\a}_{(0;\,1;\,0;\,)} =
    m_{(0;\,1;\,0;\,)}
    - \frac{1}{\left( \alpha+2 \right)}
    m_{(0;\,0;\,1;\,)}
    +	 \frac{1}{\left( \alpha+2 \right)}
    m_{(0;\,0;\,0;\,1)}
    $
\\& \\& $
	P^{\a}_{(1;\,0;\,0;\,)} =
m_{(1;\,0;\,0;\,)}
+	 \frac{1}{\left( \alpha+1 \right)}
m_{(0;\,1;\,0;\,)}
- \frac{1}{\left( \alpha+1 \right)}
m_{(0;\,0;\,1;\,)}
+	 \frac{1}{\left( \alpha+1 \right)}
m_{(0;\,0;\,0;\,1)}
$\\
&\\
\hline
&
\\
$(2|2,2,1)$ %& $P_{(0;1,1)}=m_{(0;1,1)}$ \\% & \\
 & $
P^{\a}_{(0, 0;\,0;\,1;\,1)} =
m_{(0, 0;\,0;\,1;\,1)}
-\frac{2\,} {\left( \alpha+5 \right)}
m_{(0, 0;\,0;\,0;\,1, 1)}
$
\\ & \\
 & $P^{\a}_{(0, 0;\,0;\,0;\,2)} =
m_{(0, 0;\,0;\,0;\,2)}
+	 \frac{1}{\left( \alpha+1 \right)}
m_{(1, 0;\,0;\,0;\,1)}
+	 \frac{1}{\left( \alpha+1 \right)}
$\\&$\qquad \qquad \qquad
m_{(0, 0;\,1;\,0;\,1)}
- \frac{1}{\left( \alpha+1 \right)}
m_{(0, 0;\,0;\,1;\,1)}
+	\frac{2\,} {\left( \alpha+1 \right)}
m_{(0, 0;\,0;\,0;\,1, 1)}
 $
\\ &\\&
$
P^{\a}_{(0, 0;\,0;\,2;\,)} =
m_{(0, 0;\,0;\,2;\,)}
+	 \frac{1}{\left( \alpha+1 \right)}
m_{(1, 0;\,0;\,1;\,)}
-\frac{1}{2\left( \alpha+2 \right)} ^{-1}
m_{(0, 0;\,0;\,0;\,2)}
$\\&$\qquad \qquad \qquad
-{\frac {1}{2 \left( \alpha+2 \right)  \left( \alpha+1 \right) }}
m_{(1, 0;\,0;\,0;\,1)}
- \frac{1}{\left( \alpha+1 \right)}
m_{(0, 0;\,1;\,1;\,)}
-\,{\frac {1}{2 \left( \alpha+2 \right)  \left( \alpha+1 \right) }}
m_{(0, 0;\,1;\,0;\,1)}
$\\&$\qquad \qquad \qquad
+ {\frac {2\,\alpha+5}{2 \left( \alpha+2 \right)  \left( \alpha+1 \right) }}
m_{(0, 0;\,0;\,1;\,1)}
-{\frac {1}{ \left( \alpha+2 \right)  \left( \alpha+1 \right) }}
m_{(0, 0;\,0;\,0;\,1, 1)}
$
 \\ & \\ \hline
 & \\
$(3|2,2,0)$% & $P_{(1,0;1)}=m_{(1,0;1)}$ \\ %& \\
 & $	P^{\a}_{(1, 0;\,2, 0;\,;\,)} =
m_{(1, 0;\,2, 0;\,;\,)}
+	\frac{2}{\left( \alpha+2 \right)}
m_{(1, 1;\,1, 0;\,;\,)}
+	\frac{2\,}{\left( \alpha+2 \right)}
m_{(0, 0;\,2, 1;\,;\,)}
$\\&$\qquad \qquad \qquad
+	\frac{2} {\left( \alpha+2 \right)}
m_{(0, 0;\,2, 0;\,;\,1)}
+	{\frac {\alpha+4}{ \left( \alpha+2 \right) ^{2}}}
m_{(1, 0;\,1, 0;\,;\,1)}
+	\frac{4} {\left( \alpha+2 \right)^{2}}
m_{(0, 0;\,1, 0;\,;\,1, 1)}
 $
 \\& \\& $ P^{\a}_{(0, 0;\,3, 0;\,;\,)} =
m_{(0, 0;\,3, 0;\,;\,)}
+	 \frac{1}{\left( 2\,\alpha+1 \right)}
m_{(2, 0;\,1, 0;\,;\,)}
+	\frac{2} {\left( 2\,\alpha+1 \right)}
m_{(1, 0;\,2, 0;\,;\,)}
$\\&$\qquad \qquad \qquad
+	{\frac {2}{ \left( \alpha+1 \right)  \left( 2\,\alpha+1 \right) }}
m_{(1, 1;\,1, 0;\,;\,)}
+	 \frac{1}{\left( 2\,\alpha+1 \right)}
m_{(0, 0;\,2, 1;\,;\,)}
+	\frac{2} {\left( 2\,\alpha+1 \right)}
m_{(0, 0;\,2, 0;\,;\,1)}
$\\&$\qquad \qquad \qquad
+	 \frac{1}{\left( 2\,\alpha+1 \right)}
m_{(0, 0;\,1, 0;\,;\,2)}
+	{\frac {2}{ \left( \alpha+1 \right)  \left( 2\,\alpha+1 \right) }}
m_{(1, 0;\,1, 0;\,;\,1)}
+	{\frac {2}{ \left( \alpha+1 \right)  \left( 2\,\alpha+1 \right) }}
m_{(0, 0;\,1, 0;\,;\,1, 1)}
$
\\& \\ \hline &\\
$(5|1,3,2)$ &$	P^{\a}_{(0;\,2, 1, 0;\,2, 1;\,)} =
m_{(0;\,2, 1, 0;\,2, 1;\,)}
- \frac{1}{\left( \alpha+3 \right)}
m_{(0;\,2, 1, 0;\,2, 0;\,1)}
-{\frac {\alpha+2}{ \left( \alpha+3 \right)  \left( 2\,\alpha+5 \right) }}
m_{(0;\,2, 1, 0;\,1, 0;\,2)}
$\\&$\qquad \qquad \qquad \phantom{=}
+	{\frac {1}{ \left( \alpha+3 \right)  \left( 2\,\alpha+5 \right) }}
m_{(1;\,2, 1, 0;\,1, 0;\,1)}
+	{\frac {2}{ \left( \alpha+3 \right)  \left( 2\,\alpha+5 \right) }}
m_{(0;\,2, 1, 0;\,1, 0;\,1, 1)}
$
 \\ &\\ \hline
%:
\end{tabular}
\end{center}
}
\end{table}

\end{document}